\documentclass[10pt]{article}
\usepackage{graphicx} 
\usepackage[utf8]{inputenc}
\usepackage{amsmath}
\usepackage[colorlinks,linkcolor=blue,citecolor=red]{hyperref}
\usepackage{color}
\usepackage{graphicx}
\usepackage{tikz}
\usepackage{amsfonts}
\usepackage{soul}
\usepackage{ulem}
\usepackage{cancel}
\usepackage{subcaption}
\usepackage{authblk}
\usepackage[export]{adjustbox}
\usepackage{float}

\usepackage[paperwidth=199.8mm,paperheight=287mm,centering,hmargin=25mm,vmargin=2cm]{geometry}
\usepackage{theorem}

\newcommand{\ket}[1]{\left\vert #1 \right\rangle}
\newcommand{\bra}[1]{\left\langle #1 \right\vert}

\newcommand{\ketbra}[2]{\vert#1\rangle\!\langle#2\vert}

\newcommand{\cP}[0]{\mathcal{P}}
\newcommand{\cH}[0]{\mathcal{H}}

\newcommand{\cA}[0]{\mathcal{A}}
\newcommand{\cL}[0]{\mathcal{L}}
\newcommand{\cI}[0]{\mathcal{I}}

\newcommand{\cZ}[0]{\mathcal{Z}}
\newcommand{\cX}[0]{\mathcal{X}}

\newcommand{\cN}[0]{\mathcal{N}}

\newcommand{\bC}[0]{\mathbb{C}}
\newcommand{\bR}[0]{\mathbb{R}}

\newcommand{\NP}[0]{\mbox{\bf NP}}

\newcommand{\QMA}[0]{\mbox{\bf QMA}}
\newcommand{\QCMA}[0]{\mbox{\bf QCMA}}

\newcommand{\Complete}[0]{\mbox{\bf complete}}

\newtheorem{definition}{Definition}

\newtheorem{lemma}[definition]{Lemma}

\newtheorem{fact}[definition]{Fact}
\newtheorem{remark}[definition]{Remark}
\newtheorem{theorem}[definition]{Theorem}
\newtheorem{corollary}[definition]{Corollary}

\newenvironment{proof}{{\bf Proof:}}{\hfill\rule{2mm}{2mm}}

\begin{document}
\title{Commuting Local Hamiltonian Problem on 2D beyond qubits}


\author[1]{Sandy Irani~\thanks{irani@ics.uci.edu}}
\affil[1]{Computer Science Department, University of California, Irvine, CA, USA}

\author[2]{Jiaqing Jiang~\thanks{jiaqingjiang95@gmail.com}}
\affil[2]{Computing and Mathematical Sciences, California Institute of Technology, Pasadena, CA, USA}

\maketitle
\thispagestyle{empty}


\begin{abstract}

We study the complexity of local Hamiltonians in which the terms pairwise commute. 
Commuting local Hamiltonians (CLHs) provide a way to study the role of
non-commutativity in the complexity of quantum systems and touch on many fundamental aspects of quantum computing and many-body systems, such as the quantum PCP conjecture and the area law.
Despite intense research activity since Bravyi and Vyalyi introduced the CLH problem  two
decades ago \cite{bravyi2003commutative}, its complexity remains largely unresolved; it is only known to lie in NP for a few special cases. 
Much of the recent research has focused on the physically motivated 2D case, where particles are located on vertices of a 2D grid and each term 
acts non-trivially only on the particles on a single square (or plaquette) in the lattice. In particular, Schuch~\cite{schuch2011complexity} showed that
the CLH problem on 2D with qubits is in NP.  Aharonov, Kenneth, and Vigdorovich~\cite{aharonov2018complexity}
then gave a constructive version of this result, showing an explicit algorithm to construct a ground state.
Resolving the complexity of the 2D CLH problem with higher dimensional particles has been elusive.
We prove two results for the CLH problem in 2D:
\begin{itemize}
	\item We give a non-constructive proof that the CLH problem in 2D with qutrits is in $\NP$. 
As far as we know, this is the first result for the commuting local Hamiltonian problem on 2D  beyond qubits. Our key lemma works for general qudits and might give new insights for tackling the general case. 
	\item We consider the factorized case, also studied in \cite{bravyi2003commutative}, where each term is a tensor product of single-particle Hermitian operators. We show that a factorized CLH in 2D, even on particles of arbitrary finite dimension,  is equivalent to a direct sum of qubit stabilizer Hamiltonians. This implies that the factorized 2D CLH problem is in $\NP$. This class of CLHs contains the Toric code as an example.

\end{itemize}
\end{abstract}

\pagebreak 
\thispagestyle{empty}

\tableofcontents
\thispagestyle{empty}

\pagebreak
\setcounter{page}{1}

\section{Introduction}

Understanding the properties of ground states of local Hamiltonians is a central problem in condensed matter physics.
Kitaev \cite{kitaev2003fault} famously formulated this problem as a decision problem, amenable to
analysis from the perspective of computational complexity, by
defining the local Hamiltonian problem (LHP), which asks whether the ground energy of a local Hamiltonian is below one threshold or greater than another. 
The LHP can be interpreted as the quantum generalization of the Boolean Satisfiability problem (SAT), where for SAT, all the terms $\{h_i\}_i$ are diagonal in the computational basis. 
Kitaev showed that the LHP is  $\QMA$-$\Complete$~\cite{kempe2006complexity}, a quantum analog of the  
 Cook–Levin theorem showing that SAT is $\NP$-complete. 
 Formally, a $k$-local Hamiltonian $H=\sum_i h_i$
is a Hermitian operator on $n$ qudits, where each term $h_i$ only acts on $k$ qudits. Given two parameters $a,b$ with $b-a\geq\frac{1}{poly(n)}$, the LHP is to determine whether the ground energy of $H$, namely the minimum eigenvalue, is smaller than $a$ or greater than $b$.

It is widely believed that $\QMA\neq \NP$, which would imply that the LHP is strictly harder than the SAT.
A natural question then to ask is what properties make quantum SAT (LHP) harder than classical SAT? Alternatively, what additional constraints can make LHP easier than $\QMA$-$\Complete$? 
An intermediate model that sits in between classical and quantum Hamiltonians is the Commuting Local Hamiltonian problem (CLHP),
in which the terms of the local Hamiltonian pairwise commute.
Commuting local Hamiltonians 
form an important sub-class of Hamiltonians of physical interest
and were first studied from the complexity perspective by Bravyi and Vyalyi ~\cite{bravyi2003commutative}. 
Compared to the general LHP,
the idea that CLHs should be more classical in nature stems from the
intuition~\cite{folland1997uncertainty} in quantum physics, which
suggests that it is the non-commutativity that makes the quantum world different from classical.
Moreover, since all the terms of a CLH
can be simultaneously diagonalized by a single basis, every eigenstate of the full
Hamiltonian can be specified up to degeneracies by the corresponding eigenvalue for each term.
The fact that the eigenstates have a classical specification
suggests that the structure of the eigenbasis is more classical in nature. 
However, this fact is not enough to prove CLHP is in $\NP$, since the eigenstates themselves can potentially be highly entangled, as is true for the famous example
of Kitaev's toric code \cite{toric}. Thus, a quantum witness for CLHP may still be beyond a classical description.

Commuting Hamiltonians provide a lens to study many fundamental aspects of quantum computing and many-body systems.
For example, the stabilizer framework~\cite{gottesman1997stabilizer} is the basis for most error-correcting codes, and stabilizer codes can be seen as the
ground states of commuting Hamiltonians.
Commuting Hamiltonians are commonly used as a test ground for attacking difficult problems
such as the quantum-PCP conjecture \cite{AE11, hastings2012trivial, Hastings2013}, NLTS conjecture~\cite{anshu2022nlts},    Gibbs states preparation~\cite{kastoryano2016quantum} and fast thermalization~\cite{bardet2023rapid,capel2020modified}. In particular,
commuting Hamiltonians could potentially provide insight into the area law and its connection to the efficient expressibility
of ground states. 
The Area Law states that for ground states of gapped Hamiltonians on a finite-dimensional lattice,   the entanglement entropy between two regions of ground states scales with the boundary between the regions as opposed to their volume.
The area law is known to hold for 1D Hamiltonians and is known for 2D only under the assumption that the 
Hamiltonian is {\it uniformly gapped} \cite{2Darea}.
It is widely hoped that proving the area law will lead to insight into whether such ground states can be efficiently expressed
or constructed by quantum circuits. As a subclass of LHP,  it is well known that the area law holds for CLHs, and yet, even in 2D, we do not
know whether the ground states of CLHs can be efficiently represented or constructed.

   
\subsection{Previous Results}\label{sec:prev}

Despite two decades of study, the complexity of CLHP still remains open, with a few special cases known to be in $\NP$
(\cite{bravyi2003commutative, aharonov2013commuting, aharonov2011complexity, schuch2011complexity, aharonov2018complexity, hastings2012matrix}).  
  Bravyi and Vyalyi 
 initiated this line of work, showing that
    qudit $2$-local CLHP is in $\NP$ ~\cite{bravyi2003commutative}.  
    Their proof uses a decomposition lemma based on the theory of finite-dimensional $C^*$-algebra representations.  All subsequent work on the CLHP has made use of this framework. 
    
  Aharonov and Eldar~\cite{aharonov2011complexity} extended the results to the $3$-local case for qubits and qutrits. Specifically, they proved that 3-local qubit-CLHP is in $\NP$. They also proved that  3-local qutrit CLHP is in $\NP$ on the Nearly Euclidean interaction graphs. 
  All the above results are proved by showing that there is a \textit{trivial
 ground state}, i.e. the ground state can be prepared by a constant depth quantum circuit. This constant depth circuit is the $\NP$ witness and the energy of this trivial ground state can be checked in classical polynomial time by a light-cone argument. 
  However, for 4-local CLHP, even if the interaction graph is a 2D square lattice, there are systems like the Toric code which has no trivial ground states. 
     
 In this work, we mainly focus on 4-local CLHP on a 2D square lattice with qudits (abbreviated as qudit-CLHP-2D). Specifically, consider a 2D square lattice as in Fig.~\ref{fig:introfig}(a) with a qudit $q$   on each vertex
  and on each plaquette $p$, there is a Hermitian term acting on the qudits on its four vertices. With some abuse of notations, we also use $p$ to denote the Hermitian term on the plaquette $p$. The qudit-CLHP-2D is given an $n$-qudit Hamiltonian
 $H=\sum_p p$ where $\{p\}_p$ are commuting, two parameters $a,b$ where  $b-a\geq 1/poly(n)$, decide whether the minimum eigenvalue of $H$ is smaller than $a$ or greater than $b$.
  There is an alternate 2D setting in which qudits are placed on the edges and there are Hermitian terms on ``plaquettes" and ``stars". The two settings, i.e. qudits on vertices or qudits on edges, are equivalent when the underlying graph is a 2D square lattice. ( See Appendix \ref{appendix:edge_vertex}.)
 
 \begin{figure}[htb]
    \centering
	\begin{subfigure}[b]{0.33\textwidth}
		\begin{tikzpicture}
		\draw[step=1cm,black,thick] (0,0) grid (3,3);
		\filldraw [white] (1,1) circle [radius=5pt]
		(2,1) circle [radius=5pt]
		(1,2) circle [radius=5pt]
		(2,2) circle [radius=5pt];
		\draw(1,1) node{$q_2$};
		\draw(2,1) node{$q_3$};
		\draw(1,2) node{$q_1$};
		\draw(2,2) node{$q_4$};
		\draw(1.5,1.5) node{\large $p$};
		\end{tikzpicture}	\centering\caption{}
	\end{subfigure}	\hfill
	\begin{subfigure}[b]{0.33\textwidth}
		\begin{tikzpicture}
		\draw[step=1cm,black,thick] (0,0) grid (3,3);	
		\filldraw [white] 
		(1,2) circle [radius=5pt];
		\draw(1,2) node{$q$};
		\draw(0.5,1.5) node{\large $p_1'$};
		\draw(1.5,1.5) node{\large $p_2$};
		\draw(0.5,2.5) node{\large $p_1$};
		\draw(1.5,2.5) node{\large $p_2'$};
	\end{tikzpicture}  \centering\caption{}
	\end{subfigure}\hfill
		\begin{subfigure}[b]{0.33\textwidth}
		\begin{tikzpicture}
		\draw[step=1cm,black,thick] (0,0) grid (3,3);	
		\draw [thick] 
		(0.8,1.2) circle [radius=2pt]
		(0.8,1.8) circle [radius=2pt]
		(1.2,1.2) circle [radius=2pt]
		(1.2,1.8) circle [radius=2pt]
		(2.2,1.2) circle [radius=2pt]
		(2.2,1.8) circle [radius=2pt]
		(1.8,1.2) circle [radius=2pt]
		(1.8,1.8) circle [radius=2pt];
		\draw (0.2,1.2) -- (0.7,1.2)
		(0.2,1.8) -- (0.7,1.8)
		(0.8,1.3) -- (0.8,1.7)
		(1.2,1.3) -- (1.2,1.7)
		(1.3,1.2) -- (1.7,1.2)
		(1.3,1.8) -- (1.7,1.8)
		(1.8,1.3) -- (1.8,1.7)
		(2.3,1.2) -- (2.8,1.2)
		(2.2,1.3) -- (2.2,1.7)
		(2.3,1.8) -- (2.8,1.8);
		\filldraw [white] (1,1) circle [radius=5pt]
		(2,1) circle [radius=5pt]
		(1,2) circle [radius=5pt]
		(2,2) circle [radius=5pt];
		\draw(1,1) node{\small $q_2$};
		\draw(2,1) node{\small $q_3$};
		\draw(1,2) node{\small $q_1$};
		\draw(2,2) node{\small $q_4$};
		\draw(1.5,1.5) node{\small  $p_b$};
		\draw(0.5,1.5) node{\small  $p_a$};
		\draw(2.5,1.5) node{\small  $p_c$};
	\end{tikzpicture}  \centering\caption{}
	\end{subfigure}
    	\caption{(a) Definition of the qudit-CLHP-2D. (b) The four terms which involve $q$ are $p_1,p_2,p_1',p_2'$. (c) An example of 1D structure in Schuch's~\cite{schuch2011complexity} proof. In the figure the symbol  $\circ$ represents the plaquette term acts non-trivially on the qudit. Specifically here among the 9 plaquettes in (c), only terms $p_a,p_b,p_c$ acts non-trivially on at least one of $q_1,q_2,q_3,q_4$. }\label{fig:introfig}
\end{figure}
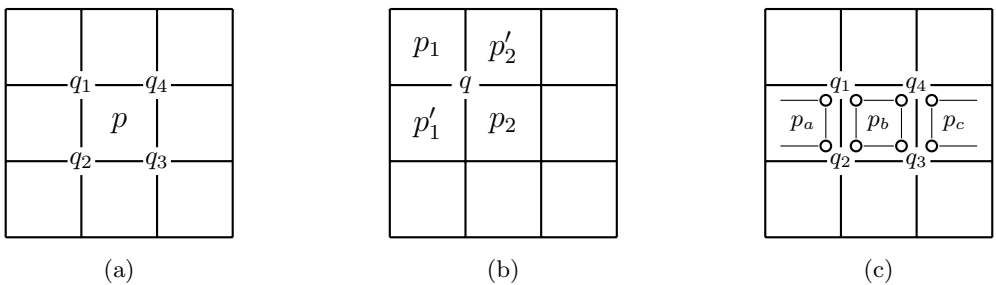
  
    Along this line, Schuch~\cite{schuch2011complexity} first proved that the 
qubit-CLHP-2D is in $\NP$.  Schuch's proof provides a witness showing that a low-energy
state exists without giving an explicit description of the state. Indeed his proof leaves open the question of whether an explicit description of the state even exists. 
 Aharonov, Kenneth, and Vigdorovich~\cite{aharonov2018complexity} later gave a constructive proof, showing that after some transformation, the qubit-CLHP-2D is equivalent to the Toric code permitting boundaries.  
  We say a proof is \textit{constructive} if 
    The $\NP$ prover  shows that the ground energy is below $a$
by providing a circuit for preparing the ground state. 
 A \textit{boundary} roughly means a qubit $q$ where one of the four terms involving this qubit, i.e. $p_1,p_2,p_1',p_2'$ in Fig.~\ref{fig:introfig}(b), 
  acts trivially on $q$. 
  A Hamiltonian is equivalent to the
 Toric code permitting boundaries if after choosing an appropriate basis for each qubit, terms $\{p\}_p$ in the Hamiltonian act similarly as $X$ or $Z$ on the non-boundary qubits. Both proofs ~\cite{schuch2011complexity,aharonov2018complexity} for the qubit-CLHP-2D heavily depend on the restrictions to the qubits, since the induced algebra   of a term $p$ on a qubit $q$ 
 can only have a very limited structure.  The limited structure for the qubits case does not hold for higher dimensional particles, not even for qutrits.

 Hastings~\cite{hastings2012matrix} proved a subclass of the qudit-CLHP-2D is in $\NP$ with some restrictive assumptions.  Roughly speaking, he grouped all qudits on the same vertical line as a supersite, then viewed the 2D lattice as a 1D line,  which reduced the qudit-CLHP-2D  back to the ``qudit'' 2-local  case~\cite{bravyi2003commutative}. However, the main problem is that those supersites are not really qudit of constant dimension. In fact, the dimension of the supersite is sub-exponential in the number of qudits. Thus Hastings further assumes that several technical conditions hold, like certain operators in the proof can be efficiently represented by matrix product operators of bounded dimension.

 Besides the qudit-CLHP-2D,  Bravyi and Vyalyi~\cite{bravyi2003commutative} also addressed the special case of the commuting Hamiltonian problem 
 where all the terms are \textit{factorized} (CHP-factorized).
 That is $H=\sum_i h_i$,
 and each term $h_i$ is a tensor product of single-qudit Hermitian operators.
 For example,  the Toric code is an instance of CHP-factorized, since
 each term is $X^{\otimes 4}$ or $Z^{\otimes 4}$. In general, in CHP-factorized each $h_i$ is not necessarily local, and there are no constraints on the underlying interaction graph.  
  Bravyi and Vyalyi give a non-constructive proof to show that the qubit-CHP-factorized is in $\NP$.    
  It is an open question whether their proof can be generalized to higher dimensional particles (qudit-CHP-factorized) or whether their proof for qubits can be made constructive.


  
 All of the above results prove that subclasses of CLHP are in $\NP$. On the other hand, Gosset, Mehta, and Vidick~\cite{gosset2017qcma} give a result that indicates CLHP might be harder than $\NP$, or even as hard as the general LHP. They
 show that the ground space connectivity problem of commuting local Hamiltonian is $\QCMA$-complete, which is as hard as the ground space connectivity problem for the general local Hamiltonian. 

 \subsection{Main results and proof overview}

\subsubsection{Main results}
In this work, we give two new results on the qudit-CLHP-2D. The family of Hamiltonians considered in both results contains the Toric code
as a special case. Theorem \ref{thm:qutrit_NP} extends the results for the qubit-CLHP-2D~\cite{schuch2011complexity,aharonov2018complexity} to qutrits.

\begin{theorem}\label{thm:qutrit_NP}
     The qutrit-CLHP-2D is in $\NP$.
\end{theorem}
 As far as we know, Theorem~\ref{thm:qutrit_NP} is the first result for CLHP on 2D lattice beyond qubits.
 As noted, the results for the qubit-CLHP-2D~\cite{schuch2011complexity,aharonov2018complexity} heavily depend on the limited dimension of qubits --- the induced algebras on qubits have a very limited structure --- which does not hold for qutrits. 
 Our key idea to circumvent this problem introduces a  technique to decrease the dimension of qudits. Specifically, denote the qudit-CLHP-2D instance as $H=\sum_p p$ and the Hilbert space of a qudit $q$ as $\cH^q$.  
 Under certain conditions, we observe that it suffices to consider a new instance of
 qudit-CLHP by (1)  Restricting $\cH^q$  to a subspace of smaller dimension and (2)  Constructing a new Hamiltonian by
 projecting all $p$ to the smaller subspace, and then rounding all
  non-1 eigenvalues of the projected terms to 0.
  Moreover, the new instance preserves the correct answer in that ``no" instances are converted to ``no" instances, and ``yes" instances are converted to ``yes" instances.
 Thus, we show that our ``decrease dimension and rounding" method (Lemma \ref{lem:decrease}) can be interpreted as a non-constructive self-reduction for the qudit-CLHP. Here self-reduction means we reduce the original qudit-CLHP to a new qudit-CLHP where some qudits have strictly smaller dimensions.
We emphasize that the key lemma, Lemma \ref{lem:decrease}, works for the qudit-CLHP rather than only for qutrits,  even without 2D geometry. This lemma might be of independent interest and bring new insights to tackle the general CLHP. We will explain this in more detail in the proof overview.

 We next consider the special case of qudit-CLHP-2D
 where all the terms are factorized (qudit-CLHP-2D-factorized).
Our second result, namely Theorem \ref{thm:fa}, is a constructive proof showing that qudit-CLHP-2D-factorized is in $\NP$.
Although we do not give the details here, our proof could be considerably simplified to provide a non-constructive
version based on the ideas of \cite{bravyi2003commutative} and \cite{schuch2011complexity}.
Here we give a stronger constructive proof that characterizes the structure of the ground space. 

\begin{theorem}[Informal version of Theorem \ref{thm:stronger}]\label{thm:fa} 
The Hamiltonian in the qudit-CLHP-2D-factorized is equivalent to a direct sum of qubit stabilizer Hamiltonian. In particular, a
factorized 2D commuting local Hamiltonian always has
 a ground state which is equivalent to qubit stabilizer state.	This implies
    qudit-CLHP-2D-factorized  is in $\NP$.  \end{theorem}
We first briefly explain terminologies in Theorem \ref{thm:fa}. 
We say a commuting Hamiltonian $H=\sum_i h_{i}$ on space $\cH_*:=\otimes_q \cH^q_*$ is equivalent to a qubit stabilizer Hamiltonian, if (1) For each $\cH^q_*$, by choosing an appropriate basis,  $\cH^q_*$ is factorized as a tensor of Hilbert spaces of dimension $2$. Thus each $\cH^q_*$ can be interpreted as several qubits. We allow $dim(\cH^q_*)=1$ which corresponds to $0$ qubit. (2) Each term $h_{i}$ acts as a Pauli operator up to phases, with respect to the basis of those ``qubits".  
We say a subspace $\cH_*\subseteq \cH:=\otimes_q \cH_q$ is \textit{simple}, if $\cH_*$  is a tensor product of subspaces of each qudit, i.e. $\cH_*=\otimes_q \cH^q_*$. 
We say a commuting Hamiltonian $H$ on $n$-qudit space $\cH=\otimes_q\cH^q$ is equivalent to a direct sum of qubit stabilizer Hamiltonian, if the $\cH$ is a direct sum of simple subspaces $\{\cH_*\}_*$, such that  $\forall i, h_i$ keeps each  $\cH_*$ invariant, and $H$
is equivalent to qubit stabilizer Hamiltonian on  $\cH_*$. 

Although one might conjecture that factorized commuting Hamiltonians (CHP-factorized) are equivalent to a direct sum of stabilizer Hamiltonians even when not restricted to 2D, this is still an open question even for the qubit case~\cite{bravyi2003commutative}.   
To illustrate the difficulty, consider two factorized terms 
acting on two qudits $q$ and $q'$: $h:=h_q\otimes h_{q'}$ and  $\hat{h}:=\hat{h}_q\otimes \hat{h}_{q'}$. 
If $h_q\hat{h}_q \neq 0$ and $h_{q'}\hat{h}_{q'} \neq 0$, then it
must be the case that the factors in each individual qudit must commute or anti-commute
which gives rise to a stabilizer-like structure.
By contrast, if $h_q\hat{h}_q=0$, then $[h,\hat{h}]=0$ for any choice of $h_{q'},\hat{h}_{q'}$.
This means that $h_{q'}$ and $\hat{h}_{q'}$ can have an arbitrary relationship to each other
and the  two commuting terms $h,\hat{h}$ may look very different from stabilizers.
 This second possibility suggests, alternatively,
 that factorized CLHs in 2D
 could have a different topological order from stabilizer Hamiltonians. More precisely, one may conjecture that there is a factorized Hamiltonian, such that no ground state can be prepared by applying a constant depth quantum circuit to a stabilizer state. 
 We show  a negative answer to this conjecture, 
 by proving that the qudit-CLHP-2D-factorized is equivalent to a direct sum of qubit stabilizer Hamiltonian.

\subsubsection{Overview for Theorem \ref{thm:qutrit_NP}.}\label{sec:overview1}

We start by reducing the more general CLH problem to a slightly restricted case where the commuting terms are projections and the energy lower bound $a$ is equal to $0$. We will argue that if the more restricted version is in $\NP$ then the more general CLHP is also in $\NP$.
Consider an instance of the more general problem with Hermitian terms $\{h_i\}_i$ and bounds $a$ and $b$, where $b -  a \ge 1/\mbox{poly}(n)$.
The $\NP$ prover can provide a vector describing the energy eigenvalue $\lambda_i$ for each individual term $h_i$ such that $\sum_{i} \lambda_i  \le a$. 
Then the verifier can replace each term $h_i$ with $\hat{h}_i = I - \Pi_{i, \lambda_i}$, where $\Pi_{i, \lambda_i}$ is the projection onto the eigenspace of $h_i$ corresponding to eigenvalue $\lambda_i$.
The new instance has a state $\ket{\psi}$ where $\hat{h}_i \ket{\psi} = 0$ for all $i$ if and only if
$h_i \ket{\psi} = \lambda_i \ket{\psi}$ for all $i$. Since the $\hat{h}_i$ are all commuting projectors, the 
eigenvalues of the Hamiltonian are non-negative integers.
Thus, the verifier can set the new $b$ to be equal to $1$, resulting in a promise gap of $1$.
Therefore, when describing the proof of Theorem \ref{thm:qutrit_NP}, we shall assume the terms are projections and that the question is whether there is a $\ket{\psi}$ where $\hat{h}_i \ket{\psi} = 0$ for all $i$
(i.e. a \textit{frustration-free} ground state). 
Note that this reduction does not work for the factorized case because even if the $h_i$ is factorized, the resulting $\hat{h}_i$ is not necessarily factorized.

We start by introducing the framework of induced algebras. The basic ideas are sketched here and
specified in more detail in Section \ref{sec:C}. 
Denote $\cL(\cH)$ as the set of all operators on a Hilbert space $\cH$. A $C^*$-algebra
is any algebra $\cA\subseteq \cL(\cH)$ which is also closed under the $\dagger$ operations and includes the identity.	Consider a Hermitian term $p$ acting $\cH\otimes \cH^c$. The operator $p$ can be decomposed
as
$$p=\sum_{i,j}  h_{ij}^{\cH} \otimes |i\rangle\langle j|^{\cH^c}.$$
The induced algebra of $p$ on space $\cH$, denoted as $\cA^\cH_p$, is the $C^*$-algebra generated by $\{I_{\cH}\}\cup  \{h_{ij}^{\cH}\}_{ij}$. Note that the particular decomposition of $p$ is not critical other
than the fact that the $|i\rangle\langle j|$ terms acting on $\cH^c$ are linearly independent.
The key technique introduced by \cite{bravyi2003commutative} is the Structure Lemma which decouples two commuting terms in their overlapping space:
consider two terms $p,\hat{p}$  sharing only one qudit $q$. If $[p,\hat{p}]=0$, then the induced algebras, $\cA^{\cH^q}_p$ and $\cA^{\cH^q}_{\hat{p}}$, must commute, meaning that every operator in $\cA^{\cH^q}_p$ commutes with
every operator in $\cA^{\cH^q}_{\hat{p}}$.  Furthermore, $p,p'$ can be decoupled in $\cH^q$, in that is there exists a direct sum decomposition $\cH^q=\bigoplus_i \cH^q_i$, where
\begin{itemize}
    \item each $\cH^q_i$ has a factorized structure:   $\cH^q_i= \cH^{q1}_i\otimes  \cH^{q2}_i$
    \item and $p$ only acts non-trivially on $\cH^{q1}_i$ and $\hat{p}$  only acts non-trivially on $\cH^{q2}_i$.
\end{itemize}

Both proofs showing  that the qubit-CLHP-2D is in $\NP$ \cite{schuch2011complexity, aharonov2018complexity}
depend heavily on the properties of qubits. In particular,
if $q$ is a qubit, then there are only two ways to have a direct sum decomposition of $\cH^q$, namely
as the direct sum of two $1$-dimensional spaces or as a single $2$-dimensional space.
Note that we must also consider the case in which an induced algebra is \textit{trivial},
meaning that $\cA^{\cH^q}_p = \{cI_{\cH}\}_c$ which implies that the operator $p$ acts trivially
on qubit $q$. 
Using the Structure Lemma the following  statement is true: 

 \begin{fact}\label{fact1}
     Any two commuting non-trivial induced algebras on a qubit must be diagonalizable in the same basis.
 \end{fact}
  Note that Fact \ref{fact1} is not true for qutrits. One may understand this statement intuitively by only basic linear algebra. In particular, 
  if $h_1,h_2,h$ are Hermitian operators on a qubit, where  $h$ is not proportional to the identity and both $h_1$ and $h_2$ commute with $h$, then all three operators can be diagonalized in the same basis. This observation is not true for qutrits, as here exist operators $h,h_1,h_2$ on a qutrit with $h$ nontrivial, such that $h$ commutes with $h_1$ and $h_2$, but the three operators cannot be diagonalized simultaneously. 

  Since the structure of our proof follows the same outline as Shuch's proof, 
we briefly explain how \cite{schuch2011complexity} used Fact \ref{fact1} to prove qubit-CLHP-2D is in $\NP$. 
Consider an arbitrary  qudit-CLHP-2D instance $H=\sum_p p$.
As argued above, we can assume $\{p\}_p$ are commuting projections, and the goal is to determine whether $\lambda(H)=0$ or $\lambda(H)\geq 1$. Note that in this setting, proving $\lambda(H)=0$ is equivalent to proving $tr(\prod_{p} (I-p))>0$.

Now consider a qubit $q$ in a 2D lattice (as shown in Fig.~\ref{fig:introfig}(b)).
We name the terms acting on $q$ as $p_1,p_2,p_1',p_2'$.
We  define the set of \textit{removable} qubits (defined implicitly in \cite{schuch2011complexity}) as follows:
\begin{itemize}
	\item $q\in R$: If the induced algebras of $p_1,p_2$  on $q$ can be diagonalized in the same basis; or this condition holds for $p_1',p_2'$.
\end{itemize}
The set $R$ is called removable since they can be effectively traced out. 
Specifically, suppose $q\in R$ and  assume $p_1,p_2$ can be diagonalized in basis $\{\ket{\phi_1},\ket{\phi_2}\}$. 
Then
$$tr\left[ \prod_{p} (I-p) \right] = \sum_{i=1,2} 
tr\left[ \ketbra{\phi_i}{\phi_i}(1 - p_1)\ketbra{\phi_i}{\phi_i} (1 - p_2) \ketbra{\phi_i}{\phi_i} \prod_{p \neq p_1, p_2} (1 - p) \right]$$
Each quantity in the sum on the right is the trace of a product of two positive semi-definite
Hermitian operators, and therefore, each quantity in the sum is non-negative.
Thus, 
 $tr(\prod_{p} (I-p))>0$ if and only if one of the two quantities in the sum is positive.
 The prover will then provide a $\ket{\phi_1}$ or $\ket{\phi_2}$ for qubit $q$ for which the trace is
 positive. Note that the proof is non-constructive because the ground state may not lie entirely
 within either the space spanned by the space spanned by $\ket{\phi_1}$ or 
 the space spanned by $\ket{\phi_2}$.
 Suppose that the witness for qubit $q$ is $\ket{\phi_1}$. The verifier must verify that
 $$ tr\left[ \ketbra{\phi_1}{\phi_1}(1 - p_1)\ketbra{\phi_1}{\phi_1} (1 - p_2) \ketbra{\phi_1}{\phi_1} \prod_{p \neq p_1, p_2} (1 - p) \right] > 0.$$
 The other two terms that might act non-trivially on qubit $q$ are $p_1'$ and $p_2'$.
By Fact \ref{fact1} we know either one of the induced algebras of $p_1', p_2'$ on $q$ acts trivially on $q$, or they can be diagonalized in the same basis.
By considering each case separately, it can be argued that the qubit $q$ can be traced out of all the 
terms. This process can be applied simultaneously for all the removable qubits. The remaining Hamiltonian will only contain
terms that operate non-trivially on qubits that are not removable.

Using Fact \ref{fact1}, we know that if $q$ is not removable ($q \not\in R$), 
then one of $p_1,p_2$ and one of $p_1',p_2'$ act trivially on $q$.
Now consider a graph where each vertex represents a plaquette term $p$ and two terms are connected
if they operate non-trivially on a common qubit.
Schuch argues that if none of the qubits are removable then this graph cannot have a vertex with a degree larger than two, namely the graph is a set of disjoint chains and cycles.
The trace of each chain or cycle can be computed in classical polynomial time
by representing each term as a tensor and contracting the tensors along the chain.

 When moving to qutrits, we need to address the fact that the Structural Lemma allows for more complex decompositions of the Hilbert space of a qutrit, and in particular, Fact \ref{fact1} no longer holds.
 Our key observation to tackle the problem, is to introduce a new way to decrease the dimension of the particle  --- \textit{Decrease Dimension and Rounding} (Lemma \ref{lem:decrease}, Sec.~\ref{sec:key}), which can be applied to qudits that have a property which we call \textit{semi-separable}.  
 
 We first describe the stronger condition of a \textit{separable} qudit, which was
  introduced in \cite{aharonov2011complexity}. Consider a CLHP $H=\sum_i h_i$.
  A qudit is separable if there exists a non-trivial decomposition $\cH^q=\bigoplus_j \cH_j^q$ such that all the terms $h_i$ keep all subspaces $\cH^q_j$ invariant. Note that if there is a separable qudit and
 a solution (i.e. a frustration-free ground state) exists, then a ground state must lie entirely
 within one of the $\cH_j^q$. Thus a prover can provide the projector $\Pi_j^q$ onto the subspace of
 qudit $q$ which contains the solution. The verifier can replace each term $h_i$ with
 $\Pi_j^q \cdot h_i \cdot \Pi_j^q$ and the dimension of the problem has been reduced. 
 
 We now extend this notion and define a qudit to be
\textit{semi-separable} if we allow at most one term to \textbf{not} keep the decomposition invariant. 
Our key observation is that, for CLHP, even for a semi-separable qudit $q$, an $\NP$ prover can similarly decrease the dimension of the qudit $q$,  in a non-constructive way. 
Specifically, the $\NP$ prover will choose a subspace $\cH_j^q$ and restrict all the terms in this subspace.
Since there is one term (call it $h_0$) that is inconsistent with the decomposition, such restriction can not be done naturally. Instead, we project $h_0$ on the subspace $\cH_j^q$ and then round all the not-1-eigenvalue to 0. By doing this we claim that we again get a  new CLHP instance with a smaller dimension in qudit $q$. More importantly, we prove that
 the original CLHP has a frustration-free ground state
 iff there exists a $j$ such that
  the new CLHP $H|_j$ also has a frustration-free ground state. 
  The reduction is non-constructive because the ground state in the new instance
  may not be the same as the ground state in the original instance. 
  The ``semi-separable" technique is powerful especially for CLHP-2D where $H=\sum_p p$, since on 2D lattice as in Fig.~\ref{fig:introfig}(b), for any qudit $q$,  if we consider the decomposition of $\cH^q$ induced by the induced algebra of $p_1$ on $q$, there are at most 2 terms, i.e. $p_1',p_2'$, which do not keep the decomposition invariant. This observation is also true for CLHP embedded on a planar graph. 
  
 By the above argument, to prove the qutrit-CLHP-2D is in $\NP$, w.l.o.g we can assume that there are no semi-separable qutrits. We further prove that the condition ---  the qutrit-CLHP-2D without semi-separable qutrits  --- leads to strong restrictions on the form of the Hamiltonian. 
 In particular, we show that for  the qutrit-CLHP-2D without semi-separable qutrits, 
 if we consider again the graph of plaquette terms where two terms are connected by an edge if they act non-trivially
 on a common qutrit,
 then this graph must also consist of disjoint chains or cycles.
 The trace of the $0$-energy space can be computed as before in classical polynomial time by contracting 1D chains of tensors.
 The $\NP$ witness will be the indexes of subspaces chosen when removing all semi-separable qudits, and the 
subspaces for the removable qutrits.

\subsubsection{Overview of Theorem \ref{thm:fa}}

Recall that Theorem  \ref{thm:fa} is a constructive proof for the qubit-CLHP-2D-factorized, where 
factorized means each term is a tensor product of single-qudit Hermitian operators.
 As we mentioned in the main results section,
 finding  a constructive proof that  CHP-factorized is in $\NP$
 is made difficult because if the product of
 two terms $h$ and $\hat{h}$ is equal to $0$,
 then their terms on individual qudits can have an arbitrary relationship.
 Namely, it is possible that $h^q \hat{h}^q \neq \pm \hat{h}^q h^q$.
On the other hand, \cite{bravyi2003commutative} showed that
 if all the terms are commuting obey the condition
 that $h^q \hat{h}^q = \pm \hat{h}^q h^q$ for each qudit,  then the Hamiltonian will be related to a qubit stabilizer Hamiltonian.

 The key part to proving the Theorem \ref{thm:fa}, is to remove the possibility that $h^q \hat{h}^q \neq \pm \hat{h}^q h^q$ for some $q$,
 without changing the ground space, which will imply a constructive proof by showing
 a correspondence with stabilizer Hamiltonians.  In general, this removal is hard to achieve. Even for the qubit-CHP-factorized, it is still an open question whether there exists a constructive proof. 
 However
 surprisingly, we can give a constructive proof for qudit-CHP-factorized, when the underlying interaction graph is 2D, i.e. qudit-CLHP-2D-factorized. Specifically, by using proof of contradiction, firstly we prove that  qudit-CLHP-2D-factorized, 
 if there are no separable qudits, then all terms must commute in a regular way.  
  With a slight clarification, we show that the qudit-CLHP-2D-factorized without separable qudits is equivalent to qubit stabilizer Hamiltonian. 
  For the more general case where there are separable qudits, we then notice that there exists a partition of the $n$-qudit space into simple subspaces, such that the restricted Hamiltonian on each subspace has no separable qudits. This partition is achieved by the following: when there is a separable qudit $q$ w.r.t decomposition $\cH^q = \oplus_j \cH^q_j$, we partition the whole space according to this decomposition. We recursively perform this partition until for each subspace,  the restricted Hamiltonian has no separable qudits.  



\subsection{Structure of the manuscript}

The manuscript is structured as follows. In Sec.~\ref{sec:prelim} we give notations and definitions which are used throughout this manuscript. In Sec.~\ref{sec:C} we review the necessary definitions and techniques of $C^*$-algebra and the Structure Lemma required for proving that the qutrit-CLHP-2D is in $\NP$. In Sec.~\ref{sec:qutrit-2D} and Sec.~\ref{sec:qudit-2D-factorized}, we give proofs for the qutrit-CLHP-2D and the qudit-CLHP-2D-factorized respectively. 

The manuscript is written in a way that several proofs can be read separately. 
In summary, for readers interested in Lemma \ref{lem:decrease} (The Decrease Dimension and Rounding Lemma), the suggested order is Sec.~\ref{sec:prelim}  and Sec.~\ref{sec:key}.  For readers interested in the qutrit-CLHP-2D, the suggested order is Sec~\ref{sec:prelim}, Sec.~\ref{sec:C}, Sec.~\ref{sec:qutrit-2D}. For readers 
 interested in the qudit-CLHP-2D-factorized, the suggested order is Sec~\ref{sec:prelim}, Sec.~\ref{sec:qudit-2D-factorized}, and then Sec.~\ref{sec:C} if necessary. 
\subsection{Conclusion and future work}

In this manuscript, we give two new results of the qudit-CLHP-2D. First, we proved that qutrit-CLHP-2D is in $\NP$,  by introducing a non-constructive way of self-reduction for the qudit-CLHP, when there are semi-separable qudits. This self-reduction (proven in Lemma \ref{lem:decrease}) works for qudit and might be of independent interest. Second, we prove that qudit-CLHP-2D-factorized is in $\NP$, by showing that the Hamiltonian is equivalent to a direct sum of qubit stabilizer Hamiltonian.

One direct question is whether our proof for qutrit-CLHP-2D can be made to be constructive. That is whether one can prepare the ground state by polynomial-size quantum circuits. Aharonov, Kenneth and Vigdorovich~\cite{aharonov2018complexity} 
proved that the qubit-CLHP-2D is equivalent to the Toric code permitting boundary. It is natural to ask whether qutrit-CLHP-2D, or general qudit-CLHP-2D, can have different ground space properties from the stabilizer Hamiltonians.  
Another question is whether our constructive proof for the qudit-CLHP-2D-factorized can be modified to prepare the ground states of the qubit-CHP (without 2D geometry). Recall that the qubit-CHP is in $\NP$ by a non-constructive method~\cite{bravyi2003commutative}.

A further question is 
 to extend the frontier of the complexity of CLHP. In particular, we conjecture that Lemma \ref{lem:decrease} can be used in more general settings. Recall that Lemma \ref{lem:decrease} works for qudit rather than only for qutrits, and it implies w.l.o.g we can assume that there are no semi-separable qudits. It is interesting to see whether one can prove that the qudit-CLHP-2D$\in\NP$ by combining Lemma \ref{lem:decrease} and other methods like \cite{aharonov2018complexity}. 
Another promising setting is considering 3-local qutrit-CLHP, without any geometry constraints. Recall that \cite{aharonov2011complexity} proved for 3-local qubit-CLHP is in $\NP$, by showing that after removing all separable qubits, the resulting Hamiltonian can be viewed as a 2-local qudit-CLHP. It is possible that for 3-local qutrit-CLHP, if we remove all semi-separable qutrits, the Hamiltonian is again of a 2-local structure, which will imply  3-local qutrit-CLHP is in $\NP$.

Most known results are trying to show that CLHP is in $\NP$. On the other side, it will be very interesting to provide any evidence that CLHP might be harder than $\NP$.

\section{Preliminaries}\label{sec:prelim}

\subsection{Notation}
Given a Hermitian operator $H$, we use $\lambda(H)$ to denote its ground energy, i.e. minimum eigenvalue. 
For two operators, $h$ and $h'$, we use $[h,h']$ to denote its commutator $hh'-h'h$. 
In particular, $[h,h']=0$ means $h,h'$ are commuting. Two sets of operators, $S$ and $\hat{S}$, commute if $[h,\hat{h}]=0,\forall h\in S,\hat{h}\in\hat{S}.$
 For a set of Hermitian operators $\{h_i\}_i$, we use $\ker\{h_i\}_i$ to denote its common $0$-eigenspace, i.e. $\ker\{h_i\}_i:=\{\ket{\psi}\, |\, h_i\ket{\psi}=0,\forall i\}$. We say $\ker\{h_i\}_i$ is non-trivial iff $\ker\{h_i\}_i\neq \{0\}.$\footnote{This $0$ means the zero vector. With some abuse of notations, we use $0$ both for real number zero, and zero vector.}  For ease of illustration, we also denote a Hermitian $H=\sum_i h_i$ as a set $\{h_i\}_i$. We say a Hermitian operator $\Pi$ is a projection if $\Pi^2=\Pi$. When $\{h_i\}_i$ are commuting projections, we have $\lambda(H)=0$ iff $\ker{\{h_i\}_i}$ is non-trivial. 
 
 Let $\cH$ be a finite-dimensional Hilbert space. We use $\cL(\cH)$ to denote the set of linear operators on $\cH$. For Hermitian $h$, we use $h\succeq 0$ to denote $h$ is positive semidefinite, that is all of its eigenvalues are non-negative. 
     We use $I$ to denote the identity matrix. Let $h$ be a Hermitian operator on Hilbert space $\cX=\cH\otimes \cZ$, we say $h$ keeps the decomposition $\cH=\bigoplus_i \cH_i$ invariant if $h$ keeps the subspace $\cH_i\otimes \cZ$ invariant, $\forall i$. We say the decomposition $\cH=\bigoplus_{i=1}^m \cH_i$ is non-trivial if $m\geq 2$.
     
     In the following,
  we use  $q$ to denote a qudit, and 
  $\cH^q$ to denote the Hilbert space of the qudit $q$. Consider an operator $h$ acting  
  on $n$ qudits. We say that 
  $h$ acts trivially on a qudit $q$ if $h$ acts as identity on  $\cH^q$.
   When $h$ acts non-trivially on  only $k$ of the $n$ qudits, we will interchangeably view $h$ as an operator on $k$ qudits or a global operator on $n$ qudits. The meaning will be clear in the context. We use $tr_S()$ for tracing out the qudits in $S$, and use $tr()$ for tracing out all the qudits. We use $S^c$ to denote the set of qudits outside $S$.
   
\subsection{Formal Problem Definitions}

\noindent\textbf{Commuting $k$-Local Hamiltonian}
We say a Hermitian operator $H$ on $n$ qudits is a commuting $k$-local Hamiltonian, if 
 $H=\sum_{i=1}^{m} h_i$ for $m=poly(n)$, where
 each $h_i$ only acts non-trivially on $k$ qudits, and  $h_ih_j=h_jh_i,\forall i,j$.
 We allow different qudits to have different dimensions. In particular, for qutrit $k$-local commuting local Hamiltonian, we allow the dimension of each qudit to be either $1,2$ or $3$. 
\vspace{0.5em}

\noindent\textbf{2D and Factorized Variants}
Consider a 2D square lattice as in Fig.~\ref{fig:introfig}(a),  on each vertex there is a qudit $q$, and on each plaquette $p$ there is a Hermitian term acting on the qudits on its four vertices. With some abuse of notations, we also use $p$ to denote the Hermitian term on the plaquette $p$. We say a commuting (4-local) Hamiltonian is on 2D if there is an underlying 2D square lattice and plaquettee terms defined as above such that $H=\sum_p p$ and all $\{p\}_p$ are pairwise commuting.


We further say a commuting (4-local) Hamiltonian on 2D is factorized, if each $p$ is factorized on its vertices, that is $p=p^{q_1}\otimes p^{q_2}\otimes p^{q_3}\otimes p^{q_4}$ for Hermitian terms $p^{q_i}$ acting on qudit $q_i$, as shown in Fig.~\ref{fig:introfig}(a). We call $p^{q_i}$ factors.
 For the Toric code, $p\in\{X^{\otimes 4},Z^{\otimes 4}\}$. \vspace{0.5em}

\noindent\textbf{Commuting $k$-Local Hamiltonian problem}
Given a  family
 of commuting $k$-local Hamiltonian $H=\sum_i h_i$ on $n$ qudits and  parameters $a,b\in \bR$ with $b-a\geq 1/poly(n)$.  The commuting $k$-local Hamiltonian problem 
w.r.t $(H,a,b)$
is a promise problem that decides whether $\lambda(H)\leq a$ or $\lambda(H)\geq b$.  
We denote this problem as qudit-CLHP$(H,a,b)$, abbreviated as qudit-CLHP when $H,a,b$ are clear in the context. Similarly, we abbreviate the 2D and 2D-factorized variants as qudit-CLHP-2D and qudit-CLHP-2D-factorized respectively.

We define a special case of the qudit-CLHP called qudit-CLHP-projection 
where each term $h_i$ is a projection,  $b=1$, and $a = 0$.
Note that since $\{h_i\}_i$ are commuting projections, we know $\lambda(H)$ must be a non-negative integer. Thus in the No instance we use $\lambda(H)\geq 1$ rather than $\lambda(H)\geq 1/poly(n)$. We define qudit-CLHP-2D-projection similarly.

\subsection{More Definitions}

Consider 
 a commuting $k$-local Hamiltonian $H=\sum_i h_i$ on  $n$ qudits with Hermitian terms $\{h_i\}_i$. Although $h_i$ acts non-trivially only on $k$ qudits, in this section we view it as an operator on $n$ qudits. 
 We will name the $n$ qudits as $q_1,q_2,...,q_n$. When we refer to an arbitrary qudit, we name it as $q$.

\begin{definition}[Separable qudit]
A qudit $q$ is separable w.r.t Hermitian terms $\{h_i\}_i$ if there exists a non-trivial decomposition of its Hilbert space $\cH^q=\bigoplus_{j=1}^m \cH^q_j$ s.t all $h_{i}$ keep the decomposition invariant. Here non-trivial means $m\geq 2$. We use $\Pi_j$ to denote the projection onto $\cH^q_j$. 
\end{definition}

The definition of separable is first introduced by \cite{aharonov2011complexity}. Roughly speaking, it says all the Hermitian terms $\{h_i\}_i$ are block-diagonalized in the same way. We introduce the notion of semi-separable qudit, which will play a key role in the proof of  the qutrit-CLHP-2D. 

\begin{definition}[Semi-separable qudit]
A qudit $q$ is semi-separable w.r.t Hermitian terms $\{h_i\}_i$  if there exists a non-trivial decomposition of its Hilbert space $\cH^q=\bigoplus_{j=1}^m \cH^q_j$ s.t all but  one $h_i$  
keeps the decomposition invariant. Here non-trivial means $m\geq 2$. We use $\Pi_j$ as the projection onto $\cH^q_j$.  By convention when referring to a specific qudit, we will denote the term which does not keep the decomposition invariant as $h_0$.
\end{definition}

Semi-separable qudit is a relaxation of separable qudits, in the sense that we allow one term to be not block-diagonalized w.r.t the decomposition $\cH^q=\bigoplus_{j=1}^m \cH^q_j$.
Note that by the definition of semi-separable, $h_i$ is Hermitian and we have $[h_i,\Pi_j]=0, \forall i\neq 0, \forall j$.
   We will repeatedly use this fact. It is also important to keep in mind that $[h_0,\Pi_j]$ might not be equal to $0$, since we allow $h_0$ not keeping $\cH^q_j$ invariant.

\section{Review of $C^*$-algebras and the Structure Lemma}\label{sec:C}

This section is a  review of $C^*$-algebra and the Structure Lemma~\cite{bravyi2003commutative}, which is the key tool to analyze the structures in the commuting local Hamiltonians.   A more detailed proof on those techniques can be seen in Sec.~7.3 of ~\cite{gharibian2015quantum}.
The following notations and lemmas are rephrased from~\cite{aharonov2018complexity} and Sec.~7.3 of ~\cite{gharibian2015quantum}. 

\subsection{Basics of $C^*$-algebras}

\begin{definition}[$C^*$-algebra] Let $\cH$ be a finite dimensional Hilbert space, a $C^*$-algebra
is any algebra $\cA\subseteq \cL(\cH)$ which is closed under the $\dagger$ operations and includes the identity.	We say that two $C^*$-algebras, $\cA$ and $\cA'$, commute if $[a,a']=0,\forall a\in\cA,a'\in\cA'.$
\end{definition}

\begin{definition}[Trivial operator and algebra]
	Let $\cH$ be a finite-dimensional Hilbert space. We say an operator $h\in\cL(\cH)$ acting trivially on $\cH$ if $h=cI_{\cH}$ for some constant $c$. We say a $C^*$-algebra on $\cA\subseteq\cL(\cH)$ is trivial if every operators in $\cA$ is trivial, i.e. $\cA=\{cI_{\cH}\}_c$. If $\cH=\cH_1\otimes \cH_2$, we say $h$ acts trivially on $\cH_1$ if $h= cI_{\cH_1}\otimes h_2$ for $h_2\in\cL(\cH_2)$.
\end{definition}

\begin{definition}[Center of $C^*$-algebra]
	The center of a $C^*$-algebra $\cA$ is defined as 
	the set of operators in $\cA$ which commutes with $\cA$, that is
	\begin{align}
		\cZ(\cA) := \{a\in \cA | [a,a']=0,\forall a'\in \cA\}.
	\end{align}	
\end{definition}

Then we introduce the induced algebra, which connects a Hermitian operator and a $C^*$-algebra. 

\begin{definition}[Induced algebra]\label{def:induced}
 Let $h$ be a Hermitian operator acting on Hilbert space $\cH\otimes\cH'$. Consider the  decomposition 
\begin{align}
	h =\sum_{i,j} h_{ij}^{\cH} \otimes |i\rangle\langle j|^{\cH'} 
\end{align}
where $\{\ket{i}\}_i$ is an orthogonal basis of $\cH'$.
The induced algebra of $h$ on $\cH$, denoted as $\cA_h^\cH$, is defined as the $C^*$-algebra generated by $\{h_{ij}^{\cH}\}_{ij}\cup\{I_{\cH'}\}$. We abbreviate $\cA^{\cH}_h$ as $\cA_h$ when $\cH$ is clear in the context. We  
abbreviate $\cA^{\cH_q}_h$ as  $\cA^{q}_h$ for qudit $q$.
\end{definition}

The induced algebra is independent of the chosen decomposition for Hermitian $h$.

\begin{lemma}[Claim B.3 of~\cite{aharonov2018complexity}]\label{lem:same}
In Definition \ref{def:induced}	consider two decompositions 
\begin{align}
	&h =\sum_{ij} h_{ij}^{\cH} \otimes g_{ij}^{\cH'}=\sum_{ij} \hat{h}_{ij}^{\cH} \otimes \hat{g}_{ij}^{\cH'}.
\end{align}
where the sets $\{g_{ij}^{\cH'}\}_{ij},\{\hat{g}_{ij}^{\cH'}\}_{ij}$ are linearly independent respectively. Then the $C^*$-algebra generated by $\{h_{ij}^{\cH}\}_{ij}$ is the same as the one generated by $\{\hat{h}_{ij}^{\cH'}\}_{ij}$.
\end{lemma}

By Lemma \ref{lem:same} we know the induced algebra of $h$ on $\cH$, i.e. $\cA^\cH_h$ in Definition \ref{def:induced}, is independent of the decomposition we choose, thus $\cA^\cH_h$  is well-defined.
Note that if there is a decomposition $\cH=\bigoplus_i\cH_i$ such that $
\cA_h^{\cH}$ keeps $\cH_i$ invariant, $\forall i$, then it follows that $h$ keeps $\cH_i$ invariant, $\forall i$.

\subsection{The Structure Lemma}

The Structure Lemma \cite{takesaki2003theory} says that every finite-dimensional $C^*$-algebra is a direct sum of algebras of all operators on a Hilbert space. 
See Sec.~7.3 of ~\cite{gharibian2015quantum} for an accessible proof of the Structure Lemma.
 The following statement is taken from \cite{aharonov2018complexity}.

\begin{lemma}[The Structure Lemma: classification of finite dimensional $C^*$-algebras~]\label{lem:C}
 Let $\cA \subseteq L(\cH)$ be a $C^*$-algebra where $\cH$ is finite dimensional. There exists a direct sum decomposition:
$
\cH =\bigoplus_i\cH_i
$
and a tensor product structure
$
\cH_i = \cH_i^1\otimes \cH_i^2
$ 
such that
$$
\cA = \bigoplus_i \cL(\cH_i^1) \otimes \cI(\cH_i^2)
$$
Furthermore,  
the center of $\cA$ is spanned by $\{\Pi_i\}_i$, where $\Pi_i$ is the projection onto the subspace $\cH_i$. 
\end{lemma}

Given a $C^*$-algebra $\cA$, we denote the decomposition $\cH=\bigoplus_i \cH_i$ in Lemma \ref{lem:C} as \textit{the decomposition induced by $\cA$}.
 Note that
here we do not argue whether the decomposition in Lemma \ref{lem:C} is unique or not. However, for clarity when we mention the decomposition induced by $\cA$, we always refer to the same canonical decomposition. For example, we can set the canonical decomposition to be the one obtained by the proof in Sec.~7.3 of ~\cite{gharibian2015quantum}.
 In the following we give some definitions of decompositions, and a further remark on Lemma \ref{lem:C}.

\begin{definition}[Trivial, Better decomposition]
    Consider the decomposition of a finite-dimensional space $\cH=\bigoplus_{i=1}^m\cH_i$. 
    We say the decomposition is trivial if $m=1$.
    We say one decomposition is 
   better\footnote{Here we measure ``better" only in terms of $m$. We do not require any relationship between the subspaces of the better decomposition $\cH=\bigoplus_{i=1}^m\cH_i$ and the worse $\cH=\bigoplus_{i=1}^{m'}\cH'_i$ for $m>m'$. {Note that even for two commuting algebras $\cA,\hat{\cA}\subseteq L(\cH)$, the two decompositions of $\cH$ induced by $\cA,\hat{\cA}$} might \textbf{not} be finer than each other. That's why we use ``better" rather than "finer" here. We define in this way just to ease notations and make our proof more precise.}
   than another if it has a bigger $m$.
\end{definition}

 Lemma \ref{lem:C} implies all operators in $\cA$ keep the decomposition $\cH=\bigoplus_i \cH_i$ invariant.  It is worth noting that the decomposition induced by $\cA$ might \textbf{not} be the best decomposition that $\cA$ keeps invariant. In particular,
 consider the $C^*$-algebra $\cA$  generated by $I$, i.e. $\{cI\}_{c\in\bC}$.  The decomposition of $\cH$ induced by $\cA$ is trivial, i.e. $\cH=\cH_1$, but $\cA$ keeps any decomposition of $\cH$ invariant. Using Lemma \ref{lem:C}, we can analyze how two induced algebras can commute with each other.

\begin{corollary}[The Structure Lemma]\label{cor:structure}	Let $\cA_h$ be a $C^*$-algebra acting on a finite dimensional $\cH$. Let $\cH=\bigoplus_i \cH_i$, $\cH_i=\cH_i^1\otimes \cH_i^2$, is the decomposition induced by $\cA_h$ by Lemma \ref{lem:C}. Consider another $C^*$-algebra $\cA_{h'}$ on $\cH$ which commutes with $\cA_{h}$, we have
$$
\begin{aligned}
	& \cA_h = \bigoplus_i \cL(\cH_i^1) \otimes \cI(\cH_i^2)\\
	& \cA_{h'} \subseteq \bigoplus_i \cI(\cH_i^1) \otimes \cL(\cH_i^2)
\end{aligned}
$$
Especially, all operators in $\cA_h,\cA_{h'}$ keep the decomposition $\cH=\bigoplus_i \cH_i$ invariant.
\end{corollary}
\begin{proof}
	Firstly by Lemma \ref{lem:C} we can get the decomposition of $\cH$ induced by $\cA_h$. Further let $\Pi_i$ be the projection onto $\cH_i$, by   Lemma \ref{lem:C} we know $\Pi_i\in\cZ(\cA_h)\subseteq \cA_h$.   Since $\cA_{h'}$ commutes with $\cA_{h}$, thus $\cA_{h'}$ commutes with $\Pi_i$, thus $\cA_{h'}$  keeps $\cH_i$ invariant. Since only $cI$ can commute with all operators in a Hilbert space, i.e. $\cL(H_i^1)$, thus we finish the proof.
\end{proof}

In the following, we give a sufficient condition that implies the decomposition of space induced by the $C^*$-algebra is non-trivial.

\begin{lemma}[Non-trivial decomposition]\label{lem:nontridec}
	Let $\cA$ be a $C^*$-algebra on a finite dimensional $\cH$. Denote the  decomposition  induced by $\cA$ in Lemma \ref{lem:C} be $\cH=\bigoplus_i \cH_i$. Consider another $C^*$-algebra $\cA'$ on $\cH$ which commutes with $\cA$. 
	If $\exists h\neq 0 \in \cA, h'\neq 0\in \cA'$ such that $hh'=0$. Then the decomposition  $\cH=\bigoplus_i H_i$ is non-trivial.
\end{lemma}
\begin{proof}
	With contradiction suppose the decomposition is 	trivial, i.e.
	 $$
	 \begin{aligned}
	 	\cH=\cH_1=\cH^1_1\otimes \cH^2_1
	 \end{aligned}
	 $$
	 By Corollary \ref{cor:structure} we have 
		$$
		\begin{aligned}
			&h\in  \cA=\cL(\cH^1_1)\otimes \cI_{\cH^2_1},\\
			& h' \in\cA'\subseteq \cI_{\cH^1_1}\otimes \cL(\cH^2_1).
		\end{aligned}
		$$
		Since $h\neq 0,h'\neq 0$, we have $hh'\neq 0$ which leads to a contradiction.
\end{proof}

\subsection{Partitions Inducted by Commuting Operators}

 The following definitions will be used throughout Sec.~\ref{sec:qudit-2D-factorized}.

\begin{definition}\label{def:comway_basic}
Let $h,h'$ be two Hermitian terms acting on $\cX\otimes\cH\otimes\cZ$ where $dim(\cH)=d$.  Suppose that  
 $h$ acts trivially on $\cZ$,  $h'$ acts trivially on $\cX$,  $[h,h']=0$, and at least one of $h,h'$ acts non-trivially on $\cH$.
  Let the decomposition $\cH=\bigoplus_i\cH_i$ be the better one induced by $\cA_h^\cH$ or $\cA_{h'}^\cH$. 
 We say that $h,h'$ commute in $(d_1,...,d_m)$-way on $\cH$ if $dim(\cH_i)=d_i$.
\end{definition}

Note that by Corollary \ref{cor:structure},
 $h,h'$ have a tensor-product structure on $\cH_i$. Since the dimension of any Hilbert space must be an integer, two terms on a qutrit $q$ of dimension $3$ can only commute in the following ways.

\begin{lemma}\label{lem:howcom}
Let $h,h'$ be two Hermitian terms acting on $\cX\otimes\cH\otimes\cZ$ where $dim(\cH)=3$.  If 
 $h$ acts trivially on $\cZ$,  $h'$ acts trivially on $\cX$,  $[h,h']=0$, and at least one of $h,h'$ acts non-trivially on $\cH$.
 Let the decomposition $\cH=\bigoplus_i\cH_i$ be the better one induced by  $\cA_h^\cH$ or $\cA_{h'}^\cH$.  
then $h,h'$ must commute on $\cH$ via one of the following ways
\begin{itemize}
	\item $(1,1,1)$-way: $\cH=\cH_1\bigoplus\cH_2\bigoplus \cH_3$, where $dim(\cH_i)=1,\forall i$.
	\item $(1,2)$-way: $\cH=\cH_1\bigoplus\cH_2$, $dim(\cH_1)=1,dim(\cH_2)=2$. 
	\item $(3)$-way: $\cH=\cH_1$, $dim(\cH_1)=3$. One of $h,h'$ acts trivially on $\cH$, and for another the induced algebra on $\cH$ is the full algebra $\cL(\cH)$.
\end{itemize}
\end{lemma}
\begin{proof}
	By corollary \ref{cor:structure},  we get a decomposition $\cH=\bigoplus_i \cH_i$  where $\cH_i=\cH^1_i\otimes \cH^{2}_i$. 
 Since the dimension of a subspace must be an integer we get the above 3 possible ways. Further for the $(3)$-way, 
 by assumption the decomposition induced by both $\cA_h^\cH,\cA_{h'}^\cH$ are $\cH=\cH_1$. Since $dim(\cH_1)=dim(\cH)=3$ is a prime, which means it can only be a tensor product of a one-dimensional Hilbert space and a three-dimensional Hilbert space. Thus for both $\cA_h^\cH,\cA_{h'}^\cH$, they equal to either $\{cI\}_c$ or $\cL(\cH)$. Since at least one of $h,h'$ acts non-trivially on $\cH$, we know one of the induced algebra are $\cL(\cH)$, w.l.o.g suppose $\cA_h^\cH=\cL(\cH)$. Again by corollary \ref{cor:structure},
$\cA_{h'}^\cH$ should be $\{cI\}_c$, thus $h'$ acts trivially on $\cH$.
\end{proof}

Similar arguments for qubits are widely used in the proof of the qubit-CLHP-2D is in $\NP$~\cite{schuch2011complexity}\cite{aharonov2018complexity}. We summarize it as below:
\begin{lemma}\label{lem:howcom2}
	If we change $dim(\cH)$ to be $2$ in the statement of Lemma \ref{lem:howcom}, then $h,h'$ must commute on $\cH$ via one of the following ways
	\begin{itemize}
		\item $(1,1)$-way if $\cH=\cH_1\bigoplus \cH_2$ where $dim(\cH_1)=dim(\cH_2)=1$.
		\item $(2)$-way if  $\cH=\cH_1$ where $dim(\cH_1)=2$.  One of $h,h'$ acts trivially on $\cH$, and for another the induced algebra on $\cH$ is the full algebra $\cL(\cH)$.
	\end{itemize}
\end{lemma}

Note that Lemma \ref{lem:howcom} and Lemma
\ref{lem:howcom2} only involve 2 commuting terms $h,h'$, and their overlapping space is only $\cH$. 
Those techniques do not directly apply to 2D Hamiltonians, where some of the terms overlap on 2 qudits.

\section{Qutrit Commuting Local Hamiltonian on 2D}\label{sec:qutrit-2D}

In this section, we will prove that the qutrit-CLHP-2D is in $\NP$. This proof is non-constructive.  Note that if the qutrit-CLHP-2D-projection is in $\NP$, then the qutrit-CLHP-2D is in $\NP$. The proof of this statement is in Appendix.~\ref{appendix:equi_proj}.  Thus 
in this section, w.l.o.g  we assume that all the terms $p$ are projections and prove that the qutrit-CLHP-2D-projection is in $\NP$.

The proof sketch is as follows. In Sec.~\ref{sec:key} we prove that we can further assume that there are no semi-separable qudits. In Sec.~\ref{sec:qutrit2D} we prove that for the qutrit-CLHP-2D without semi-separable qudits, 
    there are strong restrictions on the form of Hamiltonian. Finally,  in Sec.~\ref{sec:schuch} we prove that with such restrictions, we can use Schuch's method~\cite{schuch2011complexity} again.

\subsection{Self-reduction for CLHP with semi-separable qudits}\label{sec:key}

Lemmas in this section work for qudit-CLHP-projection\footnote{W.l.o.g we can assume that all terms are projections by Lemma \ref{lem:qudit_proj_redu}.}, that is we do not assume that each particle is a qutrit, or the Hamiltonian is on 2D.  Recall that qudit-CLHP-projection is: Consider a qudit-CLHP  $H=\sum_i h_i$ where $\{h_i\}_i$ are $k$-local projections for some constant $k$,  $[h_i,h_j]=0$ for $i\neq j$.
The question is to determine whether $\lambda(H)=0$ or $\lambda(H)\geq 1$. 
Note that $\lambda(H)=0$ iff all the commuting projections $\{h_i\}_i$ have a common $0$-eigenvector, i.e. $ker\{h_i\}_i$ is non-trivial.  When $\lambda(H)=0$, the common $0$-eigenvectors of $\{h_i\}_i$ are the ground states of $H$.  We also denote the Hamiltonian $H$ as $\{h_i\}_i$.

The key lemma in this section, Lemma \ref{lem:decrease}, is to prove that when there is a semi-separable qudit, the prover can perform a non-constructive self-reduction for the qudit-CLHP-projection. Here self-reduction means reducing the qudit-CLHP-projection to another qudit-CLHP-projection, where the Hilbert space of the qudit has a smaller dimension. Before going into the formal proofs, in the following, we intuitively explain how Lemma \ref{lem:decrease} works. Specifically, temporarily
 assume that $\lambda(H)=0$ thus we are in the Yes instance and try to prove $\lambda(H)=0$. We begin with the example when there
 is a separable qudit, then generalize this idea to the case of semi-separable qudit, and after that we give the formal proofs.

When there is a separable qudit $q$, the prover can easily perform a constructive self-reduction.
Suppose $q$ is a separable qudit,
by definition, there exists a non-trivial decomposition $\cH^q=\bigoplus_j \cH_j^q$ such that all the terms $\{h_i\}_i$ keep the decomposition invariant. Then there must be a subspace $\cH_{j_0}$ which contains a common $0$-eigenstate of $\{h_i\}_i$. Denote the projector onto $\cH_{j_0}^q$ as $\Pi_{j_0}$. The prover can give the decomposition $\cH^q=\bigoplus_j \cH_j^q$ and the index $j_0$. The verifier checks that $q$ is a separable qudit, then
restricts the space of $q$ from $\cH^q$ to $\cH_{j_0}^q$, and restricts all terms $\{h_i\}_i$ to $\{ h_i^{<j_0>}:=\Pi_{j_0}h_i\Pi_{j_0}\}_i$, and ask the prover to prove that $\{h_i^{<j_0>}\}_i$ has a common $0$-eigenstate. By definition of separable,  the decomposition is non-trivial, thus the new instance $\{h_i^{<j_0>}\}_i$ is strictly simpler in the sense that we strictly decrease the dimension of the qudit $q$. Note that this method is constructive --- the common $0$-eigenstate of the new instance $\{h_i^{<j_0>}\}_i$ is also the common $0$-eigenstate of the original instance $\{h_i\}_i$.

Our key observation is, for qudit-CLHP-projection, even for semi-separable qudit, the $\NP$ prover is able to perform a similar self-reduction, via a non-constructive way. If we follow the intuition of the separable qudit case, one might try restricting $\cH^q\rightarrow \cH^q_j$, and transform every term to be $\Pi_j h_i\Pi_j$. The problem is that since $h_0$ does not keep $\cH^q_j$ invariant, and does not commute with $\Pi_j$, the $\Pi_j h_0\Pi_j$ is no longer a projection. One may also doubt whether the resulting Hamiltonian is commuting. A more serious problem is that, unlike the case for separable qudit, since $h_0$ does not keep $\cH^q_j$ invariant, it is not clear how to connect the ground states of the original Hamiltonian to the ground states of the new Hamiltonian. 
We circumvent the problems by slightly changing the construction --- rounding $\Pi_j h_0\Pi_j$  to its $1$-eigenspace.

\begin{definition}[Reduced Hamiltonian]\label{def:reduce}
	Consider a semi-separable qudit $q$ w.r.t commuting projections $\{h_i\}_i$ and a non-trivial decomposition $\cH^q=\bigoplus_{j=1}^m \cH^q_j$, where $\Pi_j$ is the projection onto $\cH^q_j$. For any $j$, we define its $j$-th reduced Hamiltonian to be $ \{h_i^{(j)}\}_i$, or written as $H^{(j)}:=\sum_i h_i^{(j)}$, where 
	\begin{itemize}
		\item $h_i^{(j)}=\Pi_j h_i \Pi_j$, for $i\geq 1$.
		\item $h_0^{(j)}$ is the projection onto the $1$-eigenspace of $\Pi_j h_0 \Pi_j$. Assign $h_0^{(j)}$ to be $0$ when the $1$-eigenspace  is empty. It is equivalent to interpret $h^{(j)}_0$ is obtained by rounding all the strictly-smaller-than-1-eigenvalues of $\Pi_j h_0 \Pi_j$ to 0.
		\item We restrict the space of $q$ from $\cH^q$ to $\cH^q_j$. Note that all terms $h_i^{(j)}$ including $h_0^{(j)}$ keeps $\cH^q_j$ invariant, thus this restriction of space is well-defined. In summary, the original Hamiltonian $H$ acts on $\cH^q\otimes\left(\otimes_{q'\neq q} \cH^{q'}\right)$, the $j$-th reduced Hamiltonian $H^{(j)}$ acts on $\cH^q_j\otimes\left(\otimes_{q'\neq q} \cH^{q'}\right)$.
	\end{itemize}
\end{definition}
Note that the construction of reduced Hamiltonian is consistent with our previous intuition for the separable qudit  --- If $q$ is separable and $h_0$ also keeps $\cH^q_j$ invariant, then $\Pi_j h_0 \Pi_j$ is a projection thus $h_0^{(j)}=\Pi_j h_0 \Pi_j$. 
It is worth noting that the reduced Hamiltonian   keeps the ``geometry" of the original Hamiltonian:

  \begin{lemma}\label{lem:local}
  In Def.~\ref{def:reduce},
 if $h_i$ acts trivially on qudit $q'$ w.r.t space $\cH^{q'}$, then  $h_i^{(j)}$ acts trivially on qudit $q'$ w.r.t space $\cH^{q'}$ if $q'\neq q$ or $\cH^{q'}_j$ if $q'=q$ . Especially, 
  	If $H=\sum_i h_i$ is $k$-local (or on 2D), so does the $j$-th reduced Hamiltonian $H^{(j)}=\sum_i h_i^{(j)}$. 
  \end{lemma}
\begin{proof}
We prove  that if $h_0$ acts trivially on qudit $q'$, then so does 
 $h_0^{(j)}$, the proof for $i\neq 0$ is similar. By assumption,
 $h_0$ acts trivially on $q'$, thus $h_0=I_{\cH^{q'}}\otimes h$ for some projection $h$.
Recall that $q$ is the semi-separable qudit in Def.~\ref{def:reduce}. If $q'=q$, then
\begin{align}
	\Pi_j h_0\Pi_j &=\Pi_j\otimes h\\
					&= I_{\cH^{q'}_j}\otimes h.
\end{align}
 is a projection and acts trivially on $\cH_j^{q'}$.
  If $q'\neq q$,  $\Pi_j h_0\Pi_j= I_{\cH^{q'}} \otimes \Pi_j h\Pi_j$, the 1-eigenspace of $\Pi_j h_0\Pi_j$ is also of form $I_{q'}\otimes ...$ thus  acts trivially on $q'$.
\end{proof}
   
 Besides, the terms in the reduced Hamiltonian are commuting projections.
\begin{lemma}\label{lem:reduce}
	If in Def.~\ref{def:reduce}, $\{h_i\}_i$ are commuting projections, then for any $j$, the $j$-th reduced Hamiltonian $\{h^{(j)}_i\}_i$ are  commuting projections. 
\end{lemma}

\begin{proof}
Notice that, by the definition of semi-separable, we have $[h_i,\Pi_j]=0, \forall i\neq 0$. It is also important to keep in mind that $[h_0,\Pi_j]$ might not equal $0$.

Firstly we can check that all terms $\{h^{(j)}_i\}_i$ are projections. Notice that 	$h_0^{(j)}$ is a projection by definition. For  $i\neq 0$, since $h_i$ is a projection, and $[h_i,\Pi_j]=0$, we know $h^{(j)}_i:=\Pi_j h_i \Pi_j$ is a projection\footnote{The most direct way to understand this is imagining $\Pi_j,h_i$ are diagonal $0,1$ matrix, since they are commuting they can be diagonalized simultaneously.}.
	In summary, all the terms are projections.

	Then we prove that all terms are commuting. Notice that for any $i\neq 0$, for any $i'$, where $i'$ can be $0$, we have

	\begin{align}
		(\Pi_j h_i \Pi_j ) (\Pi_j h_{i'} \Pi_j ) &= (\Pi_j  \Pi_j ) (\Pi_j h_i h_{i'} \Pi_j ) \label{eq:1}\\
		& = (\Pi_j  \Pi_j ) (\Pi_j  h_{i'} h_i\Pi_j \Pi_j ) \label{eq:2}\\
		&  = (\Pi_j  \Pi_j ) (\Pi_j  h_{i'} \Pi_j h_i \Pi_j ) \label{eq:3}\\
		&  = (\Pi_j   h_{i'} \Pi_j)(\Pi_j  h_i \Pi_j )	 \label{eq:com}
	\end{align}

	where Eq.~(\ref{eq:1}) is from $[h_i,\Pi_j]=0$, Eq.~(\ref{eq:2}) is from $[h_i,h_{i'}]=0$ and $\Pi_j^2=\Pi_j$, Eq.~(\ref{eq:3}) is from $[h_i,\Pi_j]=0$. Note that we never assume that  $h_{i'}$ commutes with  $\Pi_j$.
	
	From the Eq.~(\ref{eq:com}), we know $[h_{i}^{(j)},h_{i'}^{(j)}]=0$ if $i\neq0,i'\neq 0$. Besides, we know for $i\neq 0$, $[h_{i}^{(j)},\Pi_jh_0\Pi_j]=0$.  Thus $h_{i}^{(j)}$ keeps the  $1$-eigenspace\footnote{We emphasize it is the space spanned by all $1$-eigenvector, rather than one of the eigenvector.} of $\Pi_jh_0\Pi_j$ invariant. Since $h_{i}^{(j)}$ is Hermitian, this implies $h_i^{(j)}$ commutes with the projection onto this 1-eigenspace of $\Pi_jh_0\Pi_j$, thus $[h_{i}^{(j)},h_{0}^{(j)}]=0$. In summary, all terms are commuting.
	\end{proof}
	
In summary, we have
\begin{corollary}[of Lemma \ref{lem:local} and Lemma \ref{lem:reduce}]
	If $\{ h_i\}_i$ are $k$-local (or on 2D) qudit commuting projections, then so does the $j$-th reduced Hamiltonian $\{h_i^{(j)}\}_i$.
	\end{corollary}

In addition, we give a cute lemma -- Lemma \ref{lem:rounding}. In the lemma description,
the right side of Eq.~(\ref{eq:round}) is just rounding all  non-zero coefficients $(1-\lambda)$ to $1$. 
This Lemma is simple itself but captures the key idea of ``rounding" used in Lemma \ref{lem:decrease}. It will explain why we can round all non-$1$ eigenvalue of $\Pi_j h_0\Pi_j$ to $0$, and only use the $1$-eigenspace. 

\begin{lemma}\label{lem:rounding}
	Let $f(j,\lambda)$ be a non-negative function. Then
	\begin{align}
			 \sum_{j}\sum_{\lambda\leq 1} (1-\lambda)f(j,\lambda)>0  \text{ iff } \sum_{j} \sum_{\lambda< 1} f(j,\lambda)>0.\label{eq:round}
	\end{align}
\end{lemma}
\begin{proof}
	Since $f(j,\lambda)$ is non-negative. It suffices to notice that both the left and the right inequalities are equivalent to $\exists j,\exists \lambda<1$ s.t. $f(j,\lambda)>0$.
\end{proof}

Now we are prepared to state our key lemma, which connects the original Hamiltonian to the reduced Hamiltonians.
Inspired by Schuch's idea~\cite{schuch2011complexity}, we will decompose a non-negative term into summation over many non-negative terms. However, our method here uses very different decomposition rules, 
and has key differences from his, which will be discussed in more detail after the proof.

\begin{lemma}[Decrease dimension and rounding]\label{lem:decrease}
Consider an instance of qudit-CLHP-projection on $n$ qudits, where the Hamiltonian is denoted as $H=\sum_i h_i$ for commuting projections $\{h_i\}_i$. Suppose there is a semi-separable qudit $q$ w.r.t. $\{h_i\}_i$ and non-trivial decomposition $\cH^q=\bigoplus_{j=1}^m \cH^q_j$.
For every $j$, define the $j$-th reduced Hamiltonian $H^{(j)}=\sum_i h_i^{(j)}$ as in Definition \ref{def:reduce}. 
 Then 
 \begin{align}
 	\lambda(H)=0 \text{\quad iff \quad } \exists j \text{ s.t } \lambda(H^{(j)})=0.
 \end{align}
	
\end{lemma}

\begin{proof}
Denote the $n$-qudit space as $\cH=\otimes_q \cH^q$. Define $\cH_j =  \cH^q_j \otimes_{q'\neq q} \cH^{q'}$.
 For clarity, in this proof, we use $I$ for the identity on $\cL(\cH)$. 
 When using  $tr(h)$ we always view $h$ as an operator on $\cH$, and we project out all the qudits, i.e. $tr(h):=\sum_{i}\langle i|h|i\rangle$ where $\{\ket{i}\}_i$ is the  computational basis for $\cH$. Especially, we view $\Pi_j$ as an operator in $\cH$, while view $I_{\cH_j}$ as an operator in $\cH_j$.

Note that $\{h_i\}_i$ are commuting projections,
	 proving  $\lambda(H)=0$  is equivalent to show that
	\begin{align}
		tr\left[ \prod_i (I- h_i)\right] >0.\label{eq:global}
	\end{align}
 Since $\{h_i\}_i$ are commuting, the relative order in the above formula is unimportant. Recall that $\Pi_j$ is the projection onto $\cH_j^q$.
	By assumption $\forall i\neq 0$, $h_i$  keeps $\cH_j^q$ invariant, thus
	
	\begin{align}
		tr\left[  \prod_i (I- h_i)\right] &= tr\left[\left(I-h_0\right) \prod_{i\neq 0}  \left( \sum_j \Pi_j (I- h_i) \Pi_j  \right) \right] \label{eq:dec}\\
		&= tr\left[\left(I-h_0\right)\sum_j \prod_{i\neq 0}  \left(  \Pi_j (I- h_i) \Pi_j  \right) \right] \label{eq:ort}\\
		&= \sum_j tr\left[\left(I-h_0\right)\prod_{i\neq 0}  \left(  \Pi_j (I- h_i) \Pi_j  \right) \right] 
	\end{align}
	
	Eq~(\ref{eq:dec}) is from $\sum_j \Pi_j=I$, and for $i\neq 0$, $h_i$ is Hermitian and keeps $\cH^q_j$ invariant thus $\sum_j \Pi_j h_i\Pi_j = h_i$. Eq~(\ref{eq:ort}) is from $\{\Pi_j\}_j$ are orthogonal from each other.  Besides, since $\Pi_j^2=\Pi_j$ and $tr(M\Pi_j)=tr(\Pi_jM)$ for arbitrary $M$, we have 
	
	\begin{align}
		tr\left[  \prod_i (I- h_i)\right] &= \sum_j tr\left[\left(I-h_0\right) \Pi_j \prod_{i\neq 0}  \left(  \Pi_j (I- h_i) \Pi_j  \right) \Pi_j\right] \\
		&= \sum_j tr\left[ \Pi_j\left(I-h_0\right) \Pi_j \prod_{i\neq 0}  \left(  \Pi_j (I- h_i) \Pi_j  \right)\right]\\
				&= \sum_j tr\left[ \Pi_j\left(I-\Pi_jh_0\Pi_j\right) \Pi_j \prod_{i\neq 0}  \left(  \Pi_j (I- h_i) \Pi_j  \right)\right]\label{eq:temp}
	\end{align}
	Eq.~(\ref{eq:temp}) if from $\Pi_j^2=\Pi_j$.
	Note that $\Pi_jh_0\Pi_j$ is an $n$-qudit  Hermitian operator, thus it can be diagonalized by a unitary matrix. Consider its eigenvalue decomposition,  
	denote the eigenvalues and projections onto the corresponding eigenspace as $\lambda$, $\Pi_{j,\lambda}$. That is
	
	\begin{align}
		\Pi_j h_0 \Pi_j = \sum_{\lambda} \lambda \Pi_{j,\lambda}\label{eq:eig}
	\end{align}	
	
	 	 Note that it might be possible that $\Pi_{j,\lambda}$  acts non-trivially on some $q'\neq q$ as long as $h_0$ acts non-trivially on $q'$. 
	 	Besides,  by definition of $\Pi_{j,\lambda}$ we have
	 \begin{align}
	 	\sum_{\lambda} 	\Pi_{j,\lambda} = I. \label{eq:who} 
	 \end{align}
	 
	 Since $h_0$ is a projection, we have $\Pi_j h_0 \Pi_j\succeq 0$  and $\lambda\in [0,1]$.
	 Use Eqs~(\ref{eq:eig}),(\ref{eq:who}), we have Eq.~(\ref{eq:temp}) becomes 
	 
	 \begin{align}
	 	tr\left[  \prod_i (I- h_i)\right] &= \sum_j  tr\left[ \left(\Pi_j \left( \sum_{\lambda\leq 1}(1-\lambda)\Pi_{j,\lambda}\right)\Pi_j\right) \prod_{i\neq 0}  \left(  \Pi_j (I- h_i) \Pi_j  \right)\right]\\
	 	&=\sum_j\sum_{\lambda\leq 1}\left(1-\lambda\right)tr\left[ \left(\Pi_j\Pi_{j,\lambda}\Pi_j\right)\prod_{i\neq 0}  \left(  \Pi_j (I- h_i) \Pi_j  \right)\right]\label{eq:2222}
	 \end{align}
	 
	  Define
	 
	 \begin{align}
	 	f(j,\lambda)&:=tr\left[ \left(\Pi_j\Pi_{j,\lambda}\Pi_j\right)\prod_{i\neq 0}  \left(  \Pi_j (I- h_i) \Pi_j  \right)\right]
	 \end{align}
	 
	 Since $\{\Pi_j(I-h_i)\Pi_j\}_{i\neq 0}$ are commuting projections, we know $\prod_{i\neq 0}  \left(  \Pi_j (I- h_i) \Pi_j  \right)\succeq 0$ and is Hermitian. Note that $\Pi_j\Pi_{j,\lambda}\Pi_j\succeq 0$ and is Hermitian. Since the trace of the product of two positive semi-definite Hermitian matrices is non-negative\footnote{Let $A,B$ to be arbitrary two Hermitian matrices where $A\succeq 0,B\succeq 0$. Since $A,B$ are Hermitian, consider the eigenvalue decompositions $A=\sum_i a_i |\phi_i\rangle\langle \phi_i|, a_i\geq 0$, $B=\sum_j b_j |\psi_j\rangle\langle \psi_j|, b_j\geq 0$. Then $tr(AB)=\sum_{i,j}a_ib_j |\langle \phi_i|\psi_j\rangle|^2\geq0$. },
	 we have $f(j,\lambda)$ is non-negative,
	 \begin{align}
	 	f(j,\lambda)\geq 0, \forall j,\lambda. \label{eq:nonn} 
	 \end{align}

	  	 By Lemma \ref{lem:rounding},  $tr\left[  \prod_i (I- h_i)\right]>0$ is equivalent to rounding all the non-zero coefficients in Eq.~(\ref{eq:2222}) to 1, that is equivalent as showing that
	 
	 \begin{align}
	 	&\sum_j\sum_{\lambda<1} tr\left[ \left(\Pi_j\Pi_{j,\lambda}\Pi_j\right)\label{eq:20}
	 	\prod_{i\neq 0}  \left(  \Pi_j (I- h_i) \Pi_j  \right)\right]>0\\\Leftrightarrow &\sum_j tr\left[
	 	\left(\Pi_j
	 	 \left(\sum_{\lambda<1} \Pi_{j,\lambda} \right) \Pi_j\right)
	 	 \prod_{i\neq 0}  \left(  \Pi_j (I- h_i) \Pi_j  \right)\right]>0\\
	 	 \Leftrightarrow &\sum_j tr\left[
	 	\left(\Pi_j
	 	 \left(I- h_0^{(j)} \right) \Pi_j\right)
	 	 \prod_{i\neq 0}  \left(  \Pi_j (I- h_i) \Pi_j  \right)\right]>0\label{eq:sub}\\
	 	\Leftrightarrow &\sum_j tr\left[   \left(\Pi_j - h_0^{(j)} \right) \prod_{i\neq 0}  \left(  \Pi_j (I- h_i) \Pi_j  \right)\right]>0. \label{eq:01}\\ 
	 		\Leftrightarrow & \sum_j tr\left[   \prod_{i}  \left(  \Pi_j - h^{(j)}_i  \right)\right]>0 	\label{eq:24}
	 		 	\end{align}
	 
 Eq.~(\ref{eq:sub}) is from Eq.~(\ref{eq:who}) and the definition of $h_0^{(j)}$. Eq.~(\ref{eq:01}) is from $\Pi_j h_0^{(j)}\Pi_j=h_0^{(j)}.$  Note that
 Eqs.~(\ref{eq:20}-\ref{eq:24}) and Eq.~(\ref{eq:nonn}) imply that
\begin{align}
	 tr\left[ \prod_{i}  \left(  \Pi_j - h^{(j)}_i  \right)\right] &= \sum_{\lambda<1}f(j,\lambda)\\
	 &\geq 0.
\end{align}
Thus we further have
 \begin{align}
	 	(\ref{eq:24})
	 	\Leftrightarrow & \exists j \text{ s.t } tr\left[ \prod_{i}  \left(  \Pi_j - h^{(j)}_i  \right)\right]>0 \label{eq:25}\\
	 	\Leftrightarrow & \exists j \text{ s.t } tr^{(j)}\left[ \prod_{i}  \left(  I_{\cH_j} - h^{(j)}_i  \right)\right]>0 \label{eq:28}
	\end{align}
where in Eq.~(\ref{eq:28}), the notation $tr^{(j)}$ means now we restrict the space of qudit $q$ from $\cH^q\rightarrow \cH^q_j$, and the trace is over $\cH_j$. Note that  Eq.~(\ref{eq:28}) is well defined since 
$\prod_{i}  \left(   I_{\cH_j} - h^{(j)}_i  \right)$ keeps $\cH_j$ invariant.

	By Lemma \ref{lem:reduce} we know $\{h_i^{(j)}\}$ are commuting projections on $\cH_j=\cH^q_j\otimes\left(\otimes_{q'\neq q}\cH^q\right)$. 
	 Eq.~(\ref{eq:28})  is equivalent to say $\exists j$ such that the $j$-th reduced Hamiltonian $\{h^{(j)}_i\}_i$ has a common 0-eigenvector, where the space of $q$ is $\cH_j^q$ and the $n$-qudit space is $\cH_j$.
	\end{proof}
	
	\begin{corollary}\label{cor:nosemi}
		If qudit-CLHP-projection without semi-separable qudit is in $\NP$, then qudit-CLHP-projection is in $\NP$. 
	\end{corollary}
	\begin{proof}
	Lemma \ref{lem:decrease} says when there is a semi-separable qudit, an $\NP$ prover can efficiently reduce the qudit-CLHP-projection to a new qudit-CLHP-projection by strictly decreasing the dimension of $q$. Since $dim(q)$ is a constant, an $\NP$ prover can reach to a qudit-CLHP-projection without semi-separable qudit by repeatedly performing Lemma \ref{lem:decrease} in  $poly(n)$ time.
	\end{proof}

	To help better understand Lemma \ref{lem:decrease},
	 let us discuss the differences between Lemma \ref{lem:decrease} and the method used in Schuch's paper~\cite{schuch2011complexity}. We both decompose some of the terms and get a summation of non-negative quantities, but we do such decomposition and projection in different ways. The key difference is that in our method, we guarantee all the quantities still correspond to a commuting local Hamiltonian problem. 
  Our method is more like self-reduction, and we can perform such self-reduction \textit{sequentially} until this are no semi-separable qudits.
	  On the contrary, in Schuch's method, the two projections he used for each qudit\footnote{Schuch wrote his proof in terms of the qubit, but the first decomposition part also works for qudit.}, are not commuting with each other, and the quantity does not correspond to commuting local Hamiltonian anymore, thus this decomposition technique can only be performed once rather than sequentially. 

\subsection{Restrictions on  the qutrit-CLHP-2D without semi-separable qudit}
\label{sec:qutrit2D}

From now on we will consider the 2D geometry and start our proof for the qutrit-CLHP-2D-projection is in $\NP$.	Recall that we allow qudits in the qutrit-CLHP-2D-projection to have different dimensions, i.e. either 1,2 or 3.   By Corollary \ref{cor:nosemi}, we can assume that there are no semi-separable qudits. This ``no semi-separable qudits condition", combined with the 2D geometry, will lead to strong restrictions on the form of the Hamiltonian, i.e. Lemma \ref{lem:comway}. 

We define some notations. As shown in Fig.~\ref{fig:p1234},  when considering the qutrit-CLHP-2D-projection, for a qudit $q$, 
 we use $p_1,p_2$ to denote the two plaquette projections in the diagonal direction, $p'_1,p_2'$ to denote the two plaquettes projections in the anti-diagonal direction. In the whole Sec.~\ref{sec:qutrit2D}, the relative positions of $q,p_1,p_2,p_1',p_2'$ will always obey Fig.~\ref{fig:p1234}.
 We give the following definitions.

\begin{figure}[!ht]
	\centering
 \begin{tikzpicture}
\draw[step=1cm,black,thick] (0,0) grid (3,2);
\filldraw [white] 
(1,0) circle [radius=5pt]
(2,0) circle [radius=5pt]
(1,1) circle [radius=5pt]
(1,2) circle [radius=5pt]
(2,1) circle [radius=5pt]
(2,2) circle [radius=5pt];
\draw(1,0) node{$q_4$};
\draw(2,0) node{$q_5$};
\draw(1,1) node{$q$};
\draw(2,1) node{$q_1$};
\draw(1,2) node{$q_2$};
\draw(2,2) node{$q_3$};
\draw(0.5,0.5) node{\large $p_1'$};
\draw(1.5,0.5) node{\large $p_2$};
\draw(2.5,0.5) node{\large $p_4'$};
\draw(0.5,1.5) node{\large $p_1$};
\draw(1.5,1.5) node{\large $p_2'$};
\draw(2.5,1.5) node{\large $p_3$};
\end{tikzpicture}
	\caption{Notations for  the qutrit-CLHP-2D-projection.}\label{fig:p1234}
\end{figure}
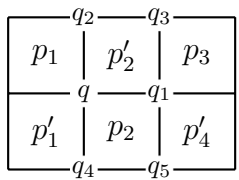

\begin{definition}\label{def:comway}
	Consider a qutrit-CLHP-2D-projection instance, for any qutrit $q$ of dimension $3$, use the notations in Fig.~\ref{fig:p1234}. For $A,B\in\{(1,1,1),(1,2),(3)\}$, we say that $p_1,p_2,p_1',p_2'$ acting on $q$ via $A\times B$-way, if $p_1,p_2$ commute in $A$-way\footnote{See Definition \ref{def:comway_basic}. More specifically, denote the set of qudits that $p_1,p_2$ acting non-trivially on as $S_1,S_2$, then in Definition \ref{def:comway_basic} $\cH:=\cH^q, \cX:=\otimes_{q'\in S_1/q}\cH^{q'},\cZ:=\otimes_{q'\in S_2/q}\cH^{q'}$.} on $\cH^q$, and $p_1',p_2'$ commute in $B$-way on $\cH^q$; or verse visa, i.e. $p_1,p_2$ commute in $B$-way on $\cH^q$, and $p_1',p_2'$ commute in $A$-way on $\cH^q$.
\end{definition}

Note that $p_1,p_2$ only overlap on one qudit -- $q$ -- thus the above sentence ``$p_1,p_2$ commute in A-way on $\cH^q$" is well-defined. Same for $p'_1,p'_2$. On the other hand, $p_1,p_1'$ overlap on two qudits, thus we cannot say  $p_1,p_1'$ commute in some way on $\cH^q$. Another clarification is, that it is possible that some of the terms are identity,  eg. $p_1=p_2=I$, then the situation does not belong to Definition \ref{def:comway}. Those cases will be considered separately and solved easily in the related proofs.

 Recall that by Corollary \ref{cor:nosemi}, we can assume that there are no semi-separable qudits. This will imply certain ways of commuting cannot exist for the qutrit-CLHP-2D-projection.

\begin{lemma}[Legal ways of commuting]\label{lem:comway}
	Consider a qutrit-CLHP-2D-projection Hamiltonian $\{p\}_p$, if there is no semi-separable qudit, then there is no qutrit $q$ with $dim(\cH^q)=3$ such that terms $p_1,p_2,p_1',p_2'$ acting on $q$ via  $(1,2)\times(3)$-way or $(1,2)\times (1,2)$-way.
\end{lemma}
\begin{proof}
To ease notations, in this proof we use $\cH$ to denote the Hilbert space of $q$, instead of using $\cH^q$. 
	With contradiction,  assume that there is a qudit $q$ with $dim(\cH^q)=3$ such that
	\begin{itemize}
		\item $p_1,p_2$ commute  via $(1,2)$-way on $\cH^q$.  The decomposition w.r.t $p_1,p_2$ is $\cH=\cH_1\bigoplus \cH_2$, where $dim(\cH_1)=1$, $\cH_1=span(\ket{\psi})$\footnote{One may wonder how could $\cH_1$ written as $\cH_1^1\otimes \cH_1^2$ in the Structure Lemma. Conceptually one can interpret $\cH_1^1=span\{\ket{\psi}\}$, and $\cH_1^2$ as a one-dimensional space as scalars $\{c\}_c$.} for some $\ket{\psi}\in\cH^q$, and $dim(\cH_2)=2$, $\cH_2=\cH_2^{1}\otimes \cH_2^{2}$. Since $2$ is a prime and $2=2\times 1$, by definition the decomposition $\cH=\cH_1\bigoplus \cH_2$ is  induced by $p_1$ or $p_2$, and $[p_1,p_2]=0$, by Corollary \ref{cor:structure},
  one can check that one of $p_1,p_2$ must acts trivially on $\cH_2$.		W.l.o.g  assume that $dim(\cH_2^{1})=2, dim(\cH_2^{2})=1$, and $p_2$ acts trivially on $\cH_2$.
		\item  For $p_1',p_2'$, consider the following cases:
			\begin{itemize}
				\item[(a)] $p_1',p_2'$ commute via $(3)$-way. In this case one of $p_1',p_2'$ must act trivially on $q$. W.l.o.g assume that $p_1'$ is  the term.
				\item[(b)] $p_1',p_2'$ commute via $(1,2)$-way. Similarly as above notaions for $p_1,p_2$ we have $\cH=\cH'_1\bigoplus \cH'_2$, and 	define $\ket{\psi'}$, $\cH_i'$ similarly. And similarly assume that $p_2'$ acts trivially on $\cH'_2$.
    \end{itemize}
	\end{itemize}

\noindent\textbf{$\square$ Case (a):} We will prove $q$ is semi-separable, which leads to a contradiction. 
Consider the decomposition from  $(1,2)$-way for $p_1,p_2$, that is $\cH=\cH_1\bigoplus \cH_2$. 
Since $p_1'$ acts trivially on $q$, it keeps those subspaces invariant as well.  In summary, all terms but $p_2'$ keep the decomposition invariant.
Since the decomposition $\cH=\cH_1\bigoplus \cH_2$ is non-trivial, by definition $q$ is semi-separable.

\vspace{0.5em}

\noindent\textbf{$\square$ Case (b):}
Consider term $p_2$, by definition, $p_2$ is Hermitian, keeps $\cH_1, \cH_2$ invariant and acts trivially on $\cH_2$. We can write 
\begin{align}
	&p_2 = \ket{\psi}\bra{\psi}\otimes A + (I_q-\ket{\psi}\bra{\psi})\otimes B.
\end{align}

Here $I_q$ is the identity on $\cH^q$. And $A,B$ (might be 0)  act non-trivially  at most on the remaining three qudits $q_1,q_4,q_5$, as in Fig.~\ref{fig:p1234}. Rewriting $h:=A-B, \hat{h}=B$, we have 
\begin{align}
	&p_2 = \ket{\psi}\bra{\psi}\otimes h + I_q\otimes \hat{h}.
\end{align}
Since $p_2$ is Hermitian, we know $h,\hat{h}$ are Hermitian.
Similarly, we can write
\begin{align}
	&p_2' = \ket{\psi'}\bra{\psi'}\otimes h' + I_q\otimes \hat{h}'.
\end{align}

where $h',\hat{h}'$ are Hermitian and  act non-trivially at most on $q_1,q_2,q_3$.  If $h=0$ then $p_2$ acts trivially on $q$. Using similar argument as case (a) we conclude $q$ is semi-separable. The case for $h'=0$ is similar. So w.l.o.g, we assume that
\begin{align}
h\neq 0,h'\neq0,\label{eq:hhn0}	
\end{align}

In the following, we are going to prove that either $q$  or $q_1$ is semi-separable, which will lead to a contradiction. W.l.o.g assume that both $\ket{\psi},\ket{\psi'}$ are unit vectors.
\begin{itemize}
	\item If $\ket{\psi}=e^{i\theta}\ket{\psi'}$ for some $\theta$: then by definition, $p_1,p_2,p_1',p_2'$ keep $\cH_1,\cH_2$ invariant. Thus $q$ is separable.
	\item If $\ket{\psi}\perp \ket{\psi'}$: then one can verify that both $p_2,p_2'$ keeps $\cH_1,\cH_2$ invariant, since $p_1$ also keeps $\cH_1,\cH_2$ invariant. By definition we know $q$ is semi-separable. 
	\item If $|\langle \psi|\psi'\rangle| \neq 0,\neq 1$: notice that	 	
		\begin{align}
			&p_2p_2' = \langle \psi|\psi'\rangle \ket{\psi}\bra{\psi'} \otimes hh'+ \ket{\psi}\bra{\psi} \otimes h\hat{h}'
			+ \ket{\psi'}\bra{\psi'} \otimes \hat{h}h'
			+ I \otimes \hat{h}\hat{h}',\\
			&p_2'p_2 = \langle \psi'|\psi\rangle \ket{\psi'}\bra{\psi} \otimes h'h+ \ket{\psi'}\bra{\psi'} \otimes h'\hat{h}
			+ \ket{\psi}\bra{\psi} \otimes \hat{h}'h
			+ I \otimes \hat{h}'\hat{h}.
		\end{align}
		
		Since $dim(q)=3$, there exists $\ket{\zeta}\neq 0\in \cH^q$ s.t $\ket{\zeta}\perp \ket{\psi},\ket{\zeta}\perp \ket{\psi'}$. By $p_2p_2'=p_2'p_2$ we have  $\langle \zeta|p_2p_2'|\zeta\rangle= \langle \zeta|p_2'p_2|\zeta\rangle$, thus 	
		\begin{align}
			\hat{h}\hat{h}'=\hat{h}'\hat{h}.\label{eq:p21}	
		\end{align}
		
		\item Since $|\langle\psi|\psi'\rangle|\neq 1$, there exists $\ket{\phi}$ such that $\ket{\phi}\perp \ket{\psi}, \ket{\phi}\not\perp \ket{\psi'}$. By Eq.~(\ref{eq:p21}) and $\langle \phi|p_2p_2'|\phi\rangle= \langle \phi|p_2'p_2|\phi\rangle$ we have
			\begin{align}
\hat{h}h'=h'\hat{h}.\label{eq:p22}	
			\end{align}
		\item Similarly to the above point, we have		
			\begin{align}
h\hat{h}'=\hat{h}'h.\label{eq:p23}	
			\end{align}
		\item Finally by Eqs.~(\ref{eq:p21})~(\ref{eq:p22})~(\ref{eq:p23}) and $p_2p_2'=p_2'p_2$ we have 
			\begin{align}
				\langle \psi|\psi'\rangle \ket{\psi}\bra{\psi'} \otimes hh' = \langle \psi'|\psi\rangle \ket{\psi'}\bra{\psi} \otimes h'h\label{eq:31}
			\end{align}	
		Left multiplying $\bra{\phi}$ to both sides of Eq.~(\ref{eq:31}), and use the fact that $\langle \psi'|\psi\rangle \neq 0, \langle \phi|\psi\rangle = 0, \langle \phi|\psi'\rangle \neq 0$, we conclude that
			\begin{align}
				h'h=0.
			\end{align}
		Similarly, we would get 
			\begin{align}
				hh'=0.
			\end{align}
		\end{itemize}	
		In summary, we showed that
		\begin{itemize}
			\item[(i)]  $h,\hat{h},h',\hat{h}'$ are Hermitian.
			\item[(ii)] $\{h,\hat{h}\}$ commute with $\{h',\hat{h}'\}$,.
			\item[(iii)] $hh'=h'h=0$.
		\end{itemize}
		To ease the notations, we abbreviate the induced algebra of Hermitian term $p$ on $q$, i.e. $\cA_p^{\cH^q}$, as   $\cA_p^q$.
		Intuitively the above $(i)(ii)(iii)$ say that  $\cA_{p_2}^{q_1}$,$\cA_{p'_2}^{q_1}$ should commute with each other, and some elements are  orthogonal to each other. 
		 We will show that this implies $q_1$ is semi-separable. In the following we write down a  careful proof,	
		 especially because $h,\hat{h}$ are  operators which might act non-trivially on three qudits $q_1,q_4,q_5$ instead of only on $q_1$. 

     Let $\{|i\rangle \langle j|^{q_4,q_5}\}_{ij}$ be the computational basis for $q_4,q_5$,   $\{|k\rangle \langle l|^{q_2,q_3}\}_{kl}$ for  $q_2,q_3$. 
		Consider the  decomposition of $h,\hat{h}$, we rewrite $p_2$ as
		$$
		\begin{aligned}
			p_2 = \ket{\psi}\bra{\psi}\otimes \sum_{ij} h^{q_1}_{ij}\otimes |i\rangle \langle j|^{q_4,q_5} +I_q\otimes \sum_{ij} \hat{h}^{q_1}_{ij}\otimes |i\rangle \langle j|^{q_4,q_5}
		\end{aligned}
		$$
		where $\{|i\rangle \langle j|^{q_4,q_5}\}_{ij}$ are  linearly independent. Since $\{\ket{\psi}\bra{\psi},I_q\}$ are linearly independent, we know $\{\ket{\psi}\bra{\psi}\otimes |i\rangle \langle j|^{q_4,q_5}\}_{ij}\cup  \{ I\otimes |i\rangle \langle j|^{q_4,q_5}\}_{ij}$ are linearly independent. By lemma \ref{lem:same}, $\cA_{p_2}^{q_1}$ is generated by $\{h^{q_1}_{ij}\}_{ij}\cup \{\hat{h}^{q_1}_{ij}\}_{ij}$.  Define similar notations for $p_2'$, we will have $\cA_{p_2'}^{q_1}$ is generated by $\{h'^{q_1}_{kl}\}_{kl}\cup \{\hat{h}'^{q_1}_{kl}\}_{kl}$.
		
		Note that $\{q_4,q_5\}$ and $\{q_2,q_3\}$ are disjoint,
		 the fact that $[h,h']=0$ implies the two sets $\{h_{ij}^{q_1}\}_{ij}$, $\{h_{kl}^{'q_1}\}_{kl}$ commute. Similarly,
		 (ii) implies the generators of $\cA_{p_2}^{q_1}$ and $\cA_{p_2'}^{q_1}$ commute, thus $\cA_{p_2}^{q_1}$ and $\cA_{p_2'}^{q_1}$ commute. If
		we name the $4$ terms on $q_1$ as $p_2',p_2,p_3,p_4'$ as in Fig.\ref{fig:p1234}.
		Let $\tilde{\cH}$ be the Hilbert space of $q_1$, consider the decomposition $\tilde{\cH}=\bigoplus_i \tilde{\cH}_i$ induced by the induced algebra of $p_2$ on $q_1$, i.e. $\cA_{p_2}^{q_1}$. By Corollary \ref{cor:structure} we know 
		all $\cA^{q_1}_p$, $p\neq p_4'$ keeps the decomposition invariant. Thus all terms expect that $p_4'$ keeps the decomposition $\tilde{\cH}=\bigoplus_i \tilde{\cH}_i$ invariant.
		
		Furthermore, $(iii)$ implies $h^{q_1}_{ij}h'^{q_1}_{kl}=0$, $\forall i,j,k,l$. By Eq.~(\ref{eq:hhn0}) we assume that $h\neq0,h'\neq 0$, thus $\exists h^{q_1}_{ij}\neq 0,h'^{q_1}_{kl}\neq 0$. Consider Lemma \ref{lem:nontridec}, let $\cA:=\cA_{p_2}^{q_1}$, $\cA':=\cA_{p'_2}^{q_1}$, we know the previous decomposition  $\tilde{\cH}=\bigoplus_i \tilde{\cH}_i$ induced by $\cA=\cA_{p_2}^{q_1}$
		 is non-trivial.

		 Combining the implications of $(i)(ii)(iii)$, we conclude $q_1$ is semi-separable, which leads to a contradiction.	
		 \end{proof}

\subsection{Schuch's method and its extensions}\label{sec:schuch}

Schuch~\cite{schuch2011complexity} proved that the qubit-CLHP-2D-projection is in $\NP$.  In this section, we illustrate that his idea can be generalized to prove that a subclass of qudit-CLHP-2D-projection is in $\NP$, see Theorem \ref{thm:remove}. In the next section, i.e. Sec.~\ref{sec:final}, we will show that  the qutrit-CLHP-2D-projection without semi-separable qudits  falls into this subclass.

The proof for  Theorem \ref{thm:remove} is similar to \cite{schuch2011complexity}. The main difference is that we generalize the definitions of removable qudits from Lemma \ref{lem:removable_2} to Lemma \ref{lem:removable_prime}. This generalization brings some subtlety so we write the proof in detail even though the proof proceeds in a similar way as \cite{schuch2011complexity}.

\begin{definition}[Removable qudit]\label{def:remove}
For qudit-CLHP-2D-projection, consider a qudit $q$ and terms $p_1,p_2,p_1',p_2'$ acting on it, as shown in Fig.~\ref{fig:p1234}. Denote $\cH^q$ as $\cH$. Suppose there exists a decomposition $\cH=\bigoplus_i \cH_i$ such that both $p_1,p_2$ keep the decomposition invariant. Let $P_i$ be the projection onto $\cH_i$. Similarly define the decomposition $\cH=\bigoplus_j \cH'_j$ and $P'_j$ for
 for $p_1',p_2'$.
    We say a qudit $q$ is removable if one of the following holds:
    \begin{itemize}
        \item[(i)]    $\forall i,j$, at most one of $p_1,p_2$ acts non-trivially on $\cH_i$, at most one of $p'_1,p'_2$ acts non-trivially on $\cH_j'$, and  $rank(P_i P^{'}_j)\leq 1$\footnote{The rank is w.r.t viewing $P_i,P_j'$ as local operators in $\cL(\cH^q)$.}.
        \item[(ii)]  $p_1',p_2'$ act trivially on $\cH$. Besides, $\cH$ has a tensor-product structure as $\cH=\hat{\cH_1}\otimes \hat{\cH_2}$, such that $p_1\in \cL(\hat{\cH_1})\otimes \cI_{\hat{\cH_2}}$,   $p_2\in \cI_{\hat{\cH_1}} \otimes \cL(\hat{\cH_2})$. Or similar conditions hold when exchanging the name of $p_1,p_2$ with $p_1',p_2'$.
    \end{itemize}
\end{definition}
We give some examples of removable qudits.

\begin{lemma}\label{lem:removable_prime}
 For qudit-CLHP-2D-projection, for any qudit $q$, if  $p_1,p_2$ commute in $(1,1,...,1)$-way on $\cH^q$, $p_1',p_2'$ commute in $(d_1',...,d_t')$-way on $\cH^q$ where $d_i'$ is a prime, $\forall i$. Then $q$ is removable.  
\end{lemma}
\begin{proof} 
Denote $\cH:=\cH^q$.
Let $\cH=\bigoplus_i \cH_i$ be the decomposition w.r.t to $p_1,p_2$ and $(1,1,...,1)$-way.  Let $\cH'=\bigoplus_i \cH'_i$ be the decomposition w.r.t to $p'_1,p'_2$ and $(d_1',...,d_t')$-way. $P_i$,$P_j'$ are the projections onto $\cH_i,\cH_j'$.
     By Definition \ref{def:comway_basic}, the way of commuting is obtained by the Structure Lemma, by Corollary \ref{cor:structure},
     we know $\cH_j'$ has a tensor-product structure as $\cH_j'=\cH_j'^1\otimes \cH_j'^2$, where 
     \begin{align}
         p_1' &= \bigoplus_i \cL(\cH^{'1}_j)\otimes {I_{\cH^{'2}_j}}\\
         p_2' &\subseteq \bigoplus_i {I_{\cH^{'1}_j}}\otimes \cL(\cH^{'2}_j)
     \end{align}
      Since $d_j'$ is a prime,  we know that at most one of $p_1',p_2'$ acts non-trivially on $\cH_j'$.  Similarly, since $1$ is a prime, at most one of $p_1,p_2$ acts non-trivially on $\cH_i$.
     Furthermore, note that $rank(P_i)=1$, $\forall i$. Thus $rank(P_iP_j')\leq \min\{rank(P_i),rank(P_j')\}=1$.
\end{proof}

\begin{lemma}
    \label{lem:removable_2}
 For the qubit-CLHP-2D-projection, for any qubit $q$, if one  of 
  $\{p_1,p_2\}$ or $\{p_1',p_2'\}$ commute in $(1,1)$-way on $\cH^q$, then $q$ is removable.   
\end{lemma}
\begin{proof}
Denote $\cH:=\cH^q$.
    W.l.o.g assume $p_1,p_2$ commute in $(1,1)$-way w.r.t $\cH=\bigoplus_i \cH_i$. If both $p_1',p_2'$ acts trivially on $q$, then $p_1',p_2'$ also keep $\cH=\bigoplus_i \cH_i$ invariant. One can check that Definition.~\ref{def:remove} (i) holds. If at least one of $p'_1,p_2'$ acts non-trivially on $\cH$, the by Lemma \ref{lem:howcom2} $p'_1,p_2'$ must commute on $\cH$ via $(1,1)$ or $(2)$-way. By Lemma \ref{lem:removable_prime}, $q$ is removable.
\end{proof}

We name those qudits as removable since we will  ``remove" them in the proof of Theorem \ref{thm:remove}. Before proving Theorem \ref{thm:remove}, we summarize \cite{schuch2011complexity}'s result as below. Although written in terms of qubit,   \cite{schuch2011complexity}'s proof directly works for the following lemma for qudits.

\begin{lemma}[One-dimensional structure~\cite{schuch2011complexity}\label{lem:1D_struc}] Consider a qudit-CLHP-2D-projection Hamiltonian $H=\sum_p p$ on $n$ qudits. Let $S$ be the set of qudits where $\forall q\in S$, there exist  $p\in\{p_1,p_2\}$ and  $p'\in\{p_1',p_2'\}$, such that both the induced algebra of $p,p'$ on $\cH^q$ are the full algebra $\cL(\cH^q)$.  Then for any quantum channels $\{\cN_p^q:\cL(\cH^q)\rightarrow \bC\}_{p,q}$, the product $\prod_p \left(\otimes_{q\in S^C}\cN_p^q\right)[I-p]$ has a one-dimensional structure thus
 $tr\left(\prod_p \left(\otimes_{q\in S^C}\cN_p^q\right)[I-p]\right)$  can be computed in classical polynomial time.
\end{lemma}
Here $\left(\otimes_{q\in S^C}\cN_p^q\right)[I-p]$ means applying 
$\otimes_{q\in S^C}\cN_p^q$ on $(I-p)$, which can be interpreted as tracing out qudits in $S^c$. 

\begin{theorem}\label{thm:remove}
    Consider a qudit-CLHP-2D-projection instance, if for each qudit $q$, either (a) $q$ is removable, or (b) there exists $p\in\{p_1,p_2\}$ and $p'\in\{p_1',p_2'\}$ such that both the induced algebra of $p,p'$  on $\cH^q$ are the full algebra $\cL(\cH^q)$, then the corresponding qudit-CLHP-2D-projection instance is in $\NP$. 
\end{theorem}
\begin{proof}
    The proof follows the idea in ~\cite{schuch2011complexity}. The main difference is that we generalize the definitions of removable qudits from Lemma \ref{lem:removable_2} to Lemma \ref{lem:removable_prime}. This generalization brings some subtlety so we write the proof in detail even though the proof proceeds in a similar way as Schuch's. 
    
    Imagine the 2D grid as a chess board and  color the plaquettes as black and white. Denote the set of the black plaquettes as $\cP_B$, the white plaquettes as $\cP_W$. 
    Use same notations as Definition \ref{def:remove}, for any removable qudit $q$, w.l.o.g assume that $p_1,p_2$ are black, $p_1',p_2'$ are white. If $q$ satisfies Definition \ref{def:remove} (i), 
     one can notice that
    \begin{align}
        &tr\left((I-p_1)(I-p_2)(I-p'_1)(I-p'_2)\right)\label{eq:main}\\ 
        &=  tr\left( \left[\sum_i P_i(I-p_1)(I-p_2)P_i\right] \left[\sum_j P^{'}_j (I-p'_1)(I-p'_2) P^{'}_j\right]
        \right)\\
        &= \sum_{i,j} tr\left( [ P_i(I-p_1)(I-p_2)P_i]  [P^{'}_j (I-p'_1)(I-p'_2) P^{'}_j]
        \right)\label{eq:subterm}
    \end{align}
    Note that by definition of removable qudit (i), at most one of $p_1,p_2$ acts non-trivially on $\cH_i$, w.o.l.g assume that it is $p_1$. This means 
    $P_i(I-p_2)P_i$ is $P_i$ tensor some operator. Formally,
        \begin{align}
        P_i (I-p_2) P_i &= tr_q(P_i (I-p_2) P_i)/tr_q(P_i) \otimes P_i\\
         &=tr_q(P_i (I-p_2) P_i)/tr_q(P_i) \otimes I_q \cdot P_i
    \end{align}
 
    Similarly, we assume that $p'_1$ acts non-trivially on $\cH'_j$.  We have
     \begin{align}
    &tr\left( P_i(I-p_1)(I-p_2)P_i P^{'}_j (I-p'_1)(I-p'_2) P^{'}_j\right)\label{eq:38} \\
    &= tr\left( P_i(I-p_1)P_i(I-p_2)P_i P^{'}_j (I-p'_1)P_j(I-p'_2) P^{'}_j\right) \\
    &= 
        tr\left( P_i(I-p_1) \cdot tr_q(P_i(I-p_2) P_i)/tr_q(P_i)\otimes I_q
        \cdot P_i \cdot P^{'}_j (I-p'_1) \cdot tr_q(P'_j (I-p'_2) P'_j)/tr_q(P'_j)\otimes I_q \cdot P^{'}_j   \right)
    \end{align}
Note that $Tr(MN)=Tr(NM)$ for any $M,N$. Further, by definition of (i), there exist un-normalized vectors $\ket{\alpha}_q,\ket{\beta}_q$ on $q$ such that $P_iP_j'=\ket{\alpha}_q\bra{\beta}_q$.  Then: 
    \begin{align}
    &= tr\left(\ket{\beta}_q \bra{\alpha }_q(I-p_1)\cdot tr_q(P_i (I-p_2) P_i)/tr_q(P_i)\otimes I_q \cdot\ket{\alpha}_q\bra{\beta}_q (I-p'_1) \cdot tr_q(P'_j (I-p'_2)P'_j)/tr_q(P'_j)\otimes I_q  \right)\\
     &= tr\left( \bra{\alpha }_q(I-p_1) \ket{\alpha}_q \cdot tr_q(P_i (I-p_2) P_i)/tr_q(P_i)\cdot \bra{\beta}_q (I-p'_1)\ket{\beta}_q  \cdot tr_q(P'_j (I-p'_2)P'_j)/tr_q(P'_j)  \right)\label{eq:proj}
    \end{align}  
     Recall that $\{p_1,p_2,p_1',p_2'\}$ are commuting projections.
    In summary, the above calculations show two things:
    \begin{itemize}
        \item In Eq.~(\ref{eq:subterm}), each quantity, i.e. Eq.~(\ref{eq:38}), is the trace of the product of two positive semi-definite Hermitian matrices, 
         thus each quantity is non-negative. Proving $tr((I-p_1)(I-p_2)(I-p_1')(I-p_2'))>0$ is equivalent to show that one of the quantitys is $>0$.
        \item In Eq.~(\ref{eq:proj}), we somehow ``project out qudit $q$" for all terms.  Then calculating Eq.~(\ref{eq:38}) is equivalent to computing the trace of a term without qudit $q$, i.e. Eq.~(\ref{eq:proj}). 
        This is why  $q$ is named removable. 
        \item Also note that in Eqs.~(\ref{eq:38})-(\ref{eq:proj}), we do not change the relative order of $(I-p_i)$ or $(I-p_i')$. This is important when considering multiple removable qudits at the same time.
    \end{itemize}
If $q$ satisfies  Definition \ref{def:remove}  (ii), we can also ``tracing out $q$''. Denote $d_q$ as the dimension of $q$. Interpret $q$ as two qudits $q_1,q_2$ w.r.t $\hat{\cH_1},\hat{\cH_2}$, we have
\begin{align}
    &tr\left((I-p_1)(I-p_2)(I-p'_1)(I-p'_2)\right)\\ &= tr\left( tr_q ((I-p_1))/{d_{q_2}} \cdot tr_q(I-p_2)/{d_{q_1}} \cdot  tr_q ((I-p'_1))/{d_q} \cdot tr_q ((I-p'_2))/{d_q} \right)\label{eq:end}
\end{align}

Eqs.~(\ref{eq:main})-(\ref{eq:end}) illustrate how to project out a single removable qudit.
Similarly, when calculating $tr(\prod_p (I-p))$ we can project out all the removable qudits. Specifically, we first perform Eq.~(\ref{eq:main})-Eq.~(\ref{eq:subterm}) simultaneously for all removable qudits, and then perform Eq.~(\ref{eq:38})-(\ref{eq:end}) to project out all removable qudits.  It is worth noting that we should be careful about the relative orders of each $(I-p)$ --- To perform the calculations in  Eq.~(\ref{eq:main})-Eq.~(\ref{eq:subterm}), one requires that for each $q$, terms $p_1,p_2$ are put in the left, and $p_1',p_2'$ in the right.\footnote{If we begin with $tr\left(P_{i_1}(I-p_1)P_{i_1} P'_{j_1}(I-p_1')P'_{j_1} P_{i_2}(I-p_2)P_{i_2} P'_{j_2}(I-p_1')P'_{j_2}\right)$ the above method doesn't work since this quantity might be negative.} To perform such decomposition for all removable qudits simultaneously, it suffices to put all the black terms on the left and all the white terms on the right. Denote the set of removable qudits as $R$, we have
    \begin{align}
        tr\left( \prod_p (I-p)\right) &= tr\left( \prod_{p\in\cP_B} (I-p) \prod_{p'\in\cP_W} (I-p')\right) \\
        & = \sum_{ i_q,j_q; q\in R} tr(\prod_{q \in R} P^q_{i_q}  
    \left(  \prod_{p\in\cP_B} (I-p) \right) \prod_{q \in R} P^q_{i_q} \prod_{q \in R} P^{'q}_{j_q}  \left( \prod_{p\in\cP_W} (I-p)\right)  \prod_{q \in R} P^{'q}_{j_q}) \label{eq:final_1D}
    \end{align}
    Then for each removable qudit $q$, we perform similar operations as Eq.~(\ref{eq:38})-(\ref{eq:end}) to project out $q$ . 
    
Finally for the quantity in Eq.~(\ref{eq:final_1D}), we project out all removable qudits for every $(I-p)$. All the remaining qudits originally correspond to type (b) in the Theorem description. By Lemma \ref{lem:1D_struc} we know  the quantity in Eq.~(\ref{eq:final_1D}) can be computed in polynomial time. 

 In summary,  proving $tr(\prod_p(I-p))>0$ is equivalent to proving that one of the quantity in Eq.~(\ref{eq:final_1D}) is $>0$, where the quantity is tracing a product which has a one-dimensional structure, thus can be computed in classical polynomial time. Thus we conclude the qudit-CLHP-2D-projection which satisfies  Theorem \ref{thm:remove}'s conditions is in $\NP$.
\end{proof}

\subsection{Qutrit-CLHP-2D is in $\NP$}\label{sec:final}

Finally,  to prove that  the qutrit-CLHP-2D is in $\NP$, it suffices to notice that the qutrit-CLHP-2D-projection without semi-separable qudit satisfying conditions in Theorem \ref{thm:remove}.

\begin{lemma}\label{lem:qutritcon}
For any  qutrit-CLHP-2D-projection instance without semi-separable qudit, every qudit satisfies one of the two conditions in Theorem \ref{thm:remove}.
\end{lemma}
\begin{proof}
	Consider any qudit $q$. If $dim(q)=1$, it is  removable due to Definition \ref{def:remove} (i). W.l.o.g. assume $dim(q)\in\{2,3\}$.
 If $p_1',p_2'$ act trivially on $\cH^q$,  since $q$ is not semi-separable, we know $p_1,p_2$ must commute via 
	$(3)$-way or $(2)$-way.	Furthermore,
	by the Structure Lemma, i.e. Corollary 
\ref{cor:structure}, we know there is a tensor product structure of $\cH^q=\cH^{q,1}\otimes \cH^{q,2}$ such  that
$p_1\in \cL(\cH^{q,1})\otimes \cI_{\cH^{q,2}}$,   $p_2\in \cI_{\cH^{q,1}} \otimes \cL(\cH^{q,2})$
, thus $q$ is removable due to Definition \ref{def:remove} (ii). A similar argument works when $p_1,p_2$ act trivially on $\cH^q$. Thus w.l.o.g, we assume that at least one of $p_1,p_2$, and one of $p_1',p_2'$ act   non-trivially on $q$. Then
	\begin{itemize} 	
		\item When $dim(q)=2$, by Lemma \ref{lem:howcom2} we know either  one of $\{p_1,p_2\}$ or $\{p_1',p_2'\}$ commute in $(1,1)$-way, or  both of them commute in $(2)$-way. For the first case,
 $q$ is  removable due to Lemma \ref{lem:removable_2}. For the second case,  $q$ satisfies 
		Theorem \ref{thm:remove} 
		condition (b).
		\item When $dim(q)=3$, since there is no semi-separable qudit, for any qudit $q$, 
  by Lemma \ref{lem:comway} and Lemma \ref{lem:howcom}, we know either  one of $\{p_1,p_2\}$ or $\{p_1',p_2'\}$ commute in $(1,1,1)$-way, then $q$ is removable by Lemma \ref{lem:removable_prime}. Or
 both of them commute in $(3)$-way, then $q$ satisfies Theorem \ref{thm:remove} condition (b).
	\end{itemize}
	\end{proof}
		
	Combined with Corollary \ref{cor:proj_redu}, Corollary \ref{cor:nosemi},  Lemma \ref{lem:qutritcon} and Theorem \ref{thm:remove}, we finally conclude that
\begin{corollary}
	The qutrit-CLHP-2D problem is in $\NP$.
\end{corollary}

\section{Factorized commuting local Hamiltonian on 2D}\label{sec:qudit-2D-factorized}

In this section, we  give a constructive proof to show that qudit-CLHP-2D-factorized is in $\NP$, by proving that qudit-CLHP-2D-factorized is equivalent to a direct sum of qubit stabilizer Hamiltonian (see Theorem~\ref{thm:stronger}).
  Note that in this section we do \textbf{not} assume that $\{p\}_p$ are projections.
 The reason for this is that if we start with an arbitrary qudit-CLHP-2D-factorized Hamiltonian, such as the Toric code, and replace each term with a projection
 that preserve the kernel (as in Lemma \ref{lem:qudit_proj_redu}), then
 the new Hamiltonian is not guaranteed to be a factorized projection. For example,
 if we take a Toric code term such as $h = XXXX$ and replace $h$ with $(IIII - XXXX)/2$,
 the resulting term is no longer factorized.
 By contrast, for  the qutrit-CLHP-2D, where we do not require that the terms be factorized, the assumption that the
 terms are projections does not weaken the results due to Corollary \ref{cor:proj_redu}.

  The structure of this section is as follows.
  In Sec.~\ref{sec:nota} we give notations and definitions. In Sec.~\ref{sec:weak} we prove that if there are no separable qudits, then the Hamiltonian is equivalent to a direct sum of qubit stabilizer Hamiltonian. Finally in Sec.~\ref{sec:strong} we prove   Theorem \ref{thm:stronger}.

\subsection{Notations, definitions and  lemmas}\label{sec:nota}

 \noindent\textbf{Notation.} 
Let $h$ be a Hermitian operator on Hilbert space $\cH$. Let $\cH'$ be a subspace of $\cH$ and suppose $h$ keeps $\cH'$ invariant. We define
  $\ker_\cH(h):=\{\ket{\psi}\in \cH \,|\, h\ket{\psi}=0\}$ and $\ker_{\cH'}(h):=\{\ket{\psi}\in \cH' \,|\, h\ket{\psi}=0\}$. For two orthogonal subspaces $V,W\subseteq \cH$, their direct sum are defined as $V\bigoplus W=\{v+w\,|\, v\in V,w\in W\}$. 
  For two Hermitian operators $h,h'\in \cL(\cH)$, we write $hh'=\pm h'h$ if either $hh'=h'h$ or $hh'=-h'h$. We say an $n$-qudit Hermitian term $h$ is factorized if
 $h=\otimes_q h^q$ where $h^q$ is Hermitian, $\forall q$. When
a factorized Hermitian 
 $h=0$, we always rewrite $h$ to be  tensor of zeros, i.e. $h^q=0,\forall q$.

We say a Hilbert space $\cH$ is equivalent to $m$-qubit space, denoted as $\cH\sim (\bC^2)^{\otimes m}$, if there exists tensor-product structure $\cH=\cH_1\otimes...\otimes \cH_m$ where $dim(\cH_i)=2$.  When considering Hilbert space $\cH\sim(\bC^{2})^{\otimes n},$ we define Pauli groups as the group of operators generated by $\{I,X,Y,Z\}^{\otimes n}$, with respect the basis which makes $\cH$ factorized as $(\bC^2)^{\otimes n}$. We denote elements in the Pauli group as Pauli operators. 
First, we  formally define what it means for a Hamiltonian to be equivalent to a qubit stabilizer Hamiltonian or a direct sum of qubit stabilizer Hamiltonians.

\begin{definition}[Equivalence to qubit stabilizer Hamiltonian]
\label{def:equi_stab}
 Consider a commuting (but not necessarily local) Hamiltonian $H=\sum_i h_i$ acting on  space $\cH_*:=\otimes_q \cH^q_*$. We say $H$ is equivalent to a qubit stabilizer Hamiltonian on $\cH_*$, if (1) For all $q$, $\cH^q_*\sim (\bC^2)^{\otimes m^q}$ for some integer $m^q$.  (2) Each $h_i$  acts as a Pauli operator up to phases on $\cH_*$, with respect to the basis which makes $\cH_*\sim (\bC^2)^{\otimes(\sum_q m^q)}$. 
\end{definition}

In the above definition, we allow $m_q=0$, where $dim(\cH^q_*)=1$, and all $h_i$ acts as $I$ up to phases on $\cH^q_*$.

\begin{definition}[Simple subspace]
	Consider an n-qudit space $\cH=\otimes_q \cH^q$. We say a subspace $\cH_*$ of $\cH$ is simple, if $\cH_*$ is a tensor product of subspace of each qudit, i.e. $\cH_*=\otimes_q \cH^q_*$ where $\cH^q_*$ is a subspace of $\cH^q$. \end{definition}
 
\begin{definition}[Equivalence to direct sum of qubit stabilizer Hamiltonian]\label{def:equi_stab_gs}
Given a commuting Hamiltonian $H=\sum_i h_i$ acting on space $\cH=\otimes_q \cH^q$. We say $H$ is equivalent to a direct sum of qubit stabilizer Hamiltonians, if there exists a set of simple subspace $\{\cH_*:=\otimes_q \cH^q_*\}_{*\in P}$ such that
  (1) $\{\cH_*\}_{*\in P}$ are pairwise orthogonal, and $\cH=\bigoplus_{*\in P} \cH_*$; (2) $\forall *\in P$, $H$ keeps $\cH_*$ invariant, $\{h_i\}_i$  keeps $\cH_*$ invariant, and $H$ is equivalent to qubit stabilizer Hamiltonian when restricted to $\cH_*.$ 
\end{definition}

\begin{remark}[Terminology of ``Equivalent to qubit stabilizer state" used in Theorem \ref{thm:fa}]\label{remark:equ_state}
	Use notations in Definition \ref{def:equi_stab_gs}, if an $n$-qudit Hamiltonian $H=\sum_i h_i$ is equivalent to a direct sum of qubit stabilizer Hamiltonians, there exists a simple subspace $\cH_*=\otimes_q\cH^q_*$ which contains a ground state of $H$, denoted as $\ket{\psi_*}$. Since $H$ is equivalent to qubit stabilizer Hamiltonian on $\cH_*$, we know $\ket{\psi_*}$ can be chosen to be a qubit stabilizer state w.r.t to the basis which makes $\cH^q_*\sim (\bC^2)^{\otimes m_q}$ in Definition \ref{def:equi_stab}. In this sense we say $H$ has a ground state which is equivalent to qubit stabilizer state. The notion of ``equivalent to qubit stabilizer state" is only used for intuitively explaining how we prove that qudit-CLHP-2D-factorized is in $\NP$ in Theorem \ref{thm:fa}.
	To avoid ambiguity, we will \textbf{not} use this notion further in the following context. \end{remark}

We now give the definitions and lemmas for commuting in a singular/regular way. For technical reasons, our definition is slightly different from \cite{bravyi2003commutative}\footnote{\cite{bravyi2003commutative} defines the case where $h\hat{h}=0$ and $\forall q, h^q\hat{h}^q=\pm \hat{h}^qh^q$ as commuting in a singular way. We define this case as commuting in a regular way.}. 

\begin{definition}\label{def:singular}[Singular/regular way]
    Consider two  factorized Hermitian terms $h,\hat{h}$ acting on  qudits, with $[h,\hat{h}]=0$. 
    \begin{itemize}
        \item   We say $h,\hat{h}$ commute in a regular way, if $\forall q$, $h^q\hat{h}^q=\pm \hat{h}^qh^q$.
        \item  We say $h,\hat{h}$ commute in a  singular way if  there $\exists$ qudit $q$, $h^q\hat{h}^q\neq\pm \hat{h}^qh^q$.
    \end{itemize} 
\end{definition}
 Recall that when
 $h=0$, we always rewrite $h$ to be  tensor of zeros. 
Thus if $h,\hat{h}$ commute in a singular way, then $h\neq 0,\hat{h}\neq 0$. We say a set of factorized Hermitian terms $\{h_i\}_i $ commute in a regular way if $\forall i,j$, the $h_i,h_j$ commute in a regular way. 
In the following, we introduce a lemma which states how factorized terms can commute with each other.  
\begin{lemma}[Rephrase of Lemma 9 in \cite{bravyi2003commutative}]\label{lem:restrictions}
    Consider two factorized Hermitian terms $h,\hat{h}$ acting on $n$ qudits, with $[h,\hat{h}]=0$. If they only share one qudit $q$, then $[h^q,\hat{h}^q]=0$. If they share two qudits $q_1,q_2$, then one of the following holds:
    (1)  $h\hat{h}\neq 0$. In this case $h^{q_1}\hat{h}^{q_1}=\pm\hat{h}^{q_1}h^{q_1}$ and $h^{q_2}\hat{h}^{q_2}=\pm\hat{h}^{q_2}h^{q_2}$.
      (2)  $h\hat{h}=0$. In this case one of $h^{q_1}\hat{h}^{q_1}$, or $h^{q_2}\hat{h}^{q_2}$ equals to 0.
\end{lemma}
\begin{corollary}\label{cor:pairox}
	Consider two factorized Hermitian terms $h,\hat{h}$ acting on $n$ qudits, with $[h,\hat{h}]=0$. If $h,\hat{h}$ share two qudits $q_1,q_2$, and commute in a singular way, then  one of $h^{q_1}\hat{h}^{q_1}$, or $h^{q_2}\hat{h}^{q_2}$ equals to 0. For the other one qudit, denoted as $q$, we have $h^q\hat{h}^q\neq\pm \hat{h}^qh^q$.
\end{corollary}

We also prove  some useful lemmas.

\begin{lemma}\label{lem:alpha0}
    Consider two Hermitian terms $h,\hat{h}$ acting on a Hilbert space $\cH$. If $h\hat{h}=\alpha \hat{h}h$ for some $\alpha\in \bR$, then $h$ keeps $\ker_{\cH}(\hat{h})$ invariant.
\end{lemma}
\begin{proof}
    It suffices to notice that $\forall \ket{\psi}\in \ker_{\cH}(\hat{h})$, we have $\hat{h}h(\ket{\psi})=\alpha h \hat{h}\ket{\psi}=0$. This implies that $h\ket{\psi}\in \ker_{\cH}(\hat{h})$. 
\end{proof}

\begin{lemma}\label{lem:factorize_separable}
Consider a qudit $q$ and Hermitian terms $\hat{h}^q\neq 0$ and  $\{h_i^q\}_i$ acting on $\cH^q$. Suppose that  $\forall i$, $\hat{h}^qh_i^q=\pm h_i^q\hat{h}^q$, and furthermore  there exists $i_0$  such that $h_{i_0}^q\neq 0$, and $\hat{h}^qh_{i_0}^q=0.$ Then $q$ is separable with respect to $\{\hat{h}^q\}\cup \{h_i^q\}_i$.
\end{lemma}
\begin{proof}
    Define $\cH_1:=\ker(\hat{h})$. Since both $\hat{h},h_{i_0}$ are non-zero and $\hat{h}h_{i_0}=0$, we know $\cH_1\neq \{0\}$ and $\cH_1\neq \cH^q.$  This implies the decomposition $\cH^q =\cH_1\bigoplus \cH_1^\perp$ is non-trivial.
    By lemma \ref{lem:alpha0}, we know that all operators in $\{\hat{h}\}\cup\{h_i\}_i$ keep $\cH_1$ invariant. Since 
    they are Hermitian,  they also keep $\cH_1^\perp$ invariant. 
    Thus $q$ is separable.
\end{proof}

\subsection{A weaker version without separable qudits}\label{sec:weak}

In this section we  prove in Theorem \ref{thm:weaker} that  if there are no separable qudits, then qudit-CLHP-2D-factorized is equivalent to qubit stabilizer Hamiltonian. In particular, we prove that in this situation, all the terms must commute in a regular way. 

\begin{lemma}\label{lem:all_regular}
    Consider qudit-CLHP-2D-factorized Hamiltonian $H=\sum_p p$ acting on $n$ qudits. If there are no separable qudits, then all the terms $\{p\}_p$ commute in a regular way. Moreover, for any  qudit $q$ and any term  $p,\hat{p}\in\{p_1,p_2,p_1',p_2'\}$ acting on $q$, as shown in Fig.~\ref{fig:pattern} (a), we have (1) $(p^q)^2=c_{pq} I_q$ for some constant $c_{pq}$. (2) $p^q\hat{p}^q=\pm \hat{p}^qp^q$. (3) The $C^*$-algebra generated by $\{p^q\}_{p\in\{p_1,p_2,p_1',p_2'\}}$ is the whole algebra $\cL(\cH^q)$.
\end{lemma}
\begin{proof}
  We first define some notations to help the illustration.  To match the notations in Definition \ref{def:equi_stab}, we denote the Hilbert space of qudit $q$ as $\cH_*^q$, the $n$-qudit space as $\cH_*:=\otimes_q \cH_*^q$.  As shown in Fig.~\ref{fig:pattern} (a), Consider a qudit $q$,  we denote the terms acting on $q$ as $p_1,p_2,p_1',p_2'$.   Recall that all the terms are factorized, thus we can use $p_1^q$ to represent the factor of $p_1$ on $q$. Use similar notations for other terms.   For any two terms $h,\hat{h}\in\{p_1,p_2,p'_1,p'_2\}$,   we use symbols 
$-,\underline{0},\times$ to represent the relationship between $h^q,\hat{h}^q$. 
\begin{itemize}
    \item ``$-$": If $h^q\hat{h}^q=\pm \hat{h}^qh^q$, and  $h,\hat{h}$ commute in a regular way.
    \item ``$\underline{0}$": If $h^q\hat{h}^q=0$, and  $h,\hat{h}$ commute in a  singular way.
    \item ``$\times$": If $h^q\hat{h}^q\neq\pm\hat{h}^qh^q$,  and $h,\hat{h}$ commute in a  singular way.
\end{itemize}

Using the above symbols, we will draw a graph to represent the relationship of $\{p_1,p_2,p'_1,p'_2\}$ on qudit $q$.
For example, the graph in Fig.~\ref{fig:pattern} (f) means: The relationship between $p_1^q$ and $p_2^q$ is ``$-$". That is
$p_1^qp_2^q=\pm p_2^qp_1^q$, and  $p_1,p_2$ commute in a regular way. Similar   for  $p^{'q}_1$ and $p^{'q}_2$. The relationship between $p^q_1$ and $p^{'q}_1$ is 
``$\underline{0}$". That is
$p_1^qp_1^{'q}=0$, $p_1,p_1'$ commute in a  singular way. 
Similar   for  $p^{'q}_2$ and $p^{q}_2$. The relationship for $p_1^q,p_2^{'q}$ is
``$\times$". That is
$p_1^qp_2^{'q}\neq\pm p_2^{'q}p_1^q$, and  $p_1,p_2'$ commute in a  singular way. Similar for $p_1^{'q}$ and $p_2^q$.
     We use $\#_q \underline{0}$ to represent the number of $\underline{0}$ in the graph for $q$, and similar for  symbol ``$\times$". For example, in Fig.~\ref{fig:pattern} (f), we have $\#_q \underline{0}=\#_q \times =2.$
     
       For each $q$, we draw such a graph and assign $-,\underline{0},\times$  to each pair $h,\hat{h}\in\{p_1^q,p_2^q,p_1^{'q},p_2^{'q}\}$. 
  For $q$ in the boundary\footnote{Here we mean the physical boundary of the lattice, not the (virtual) boundary defined in \cite{aharonov2018complexity}. } of the lattice, some terms might be missing. We  draw the graph for the existing terms similarly.
We use $\# \underline{0}$ to represent the total number of $\underline{0}$ when considering all qudits, that is $\# \underline{0}=\sum_q \#_q \underline{0}$. Similarly for $\# \times$. Note that by Corollary \ref{cor:pairox},
we  have $\# \underline{0}=\#\times$.
Now we are prepared to prove  Lemma \ref{lem:all_regular}. 
We first prove:

\vspace{0.3em}

\noindent\textbf{Claim 1: There is no qudit $q$ such that two of the terms acting on $q$, i.e. two of $p_1,p_2,p_1',p_2'$ in Fig.~\ref{fig:pattern}(a), can 
 commute in a  singular way.} With contraction suppose there are such qudits, then $\# \underline{0}=\#\times\geq 1$.
Then there must be a qudit $q$ such that  of $\#_q \underline{0}\geq \#_q\times\geq 1$. We will prove that $q$ is separable thus leading to a contradiction.
Here we suppose $q$ is in not on the boundary, thus  there are four terms $p_1,p_2,p_1',p_2'$ acting on $q$. The case where $q$ is on the boundary and some terms are missing can be analyzed similarly.

 Note that we always have $p_1^qp_2^q=p_2^qp_1^q$ since $p_1,p_2$ only shares one qudit. Similar for $p_1^{'q}$ and $p_2^{'q}$.
One can check that up to rotating the lattice, the  relationships of terms on $q$ must be related to one of the graphs in Fig.~\ref{fig:pattern} (b-f).
Please read the captions of Fig.~\ref{fig:pattern} for a precise description of the classification.
 Note that if two terms $h,\hat{h}$ on $q$ satisfies $\underline{0}$, 
by definition of commuting in a singular way, we have $p\neq 0,\hat{p}\neq 0$, thus $p^q\neq 0,\hat{p}^q\neq 0$.

For Fig.~\ref{fig:pattern} (b-e), 
let $\hat{h}^q:=p_1^q$, $\{h_i^q\}_i:=\{p_2^q,p_1^{'q},p_2^{'q}\}$. By Lemma \ref{lem:factorize_separable}, we know $q$ is separable. For Fig.~\ref{fig:pattern} (f), let $\cA$ to be the $C^*$-algebra generated by $p_1^q,p_2^{'q}$, $\cA'$ to be the $C^*$-algebra generated by $p_1^{'q},p_2^{q}$. By Lemma \ref{lem:nontridec} we know $q$ is separable.

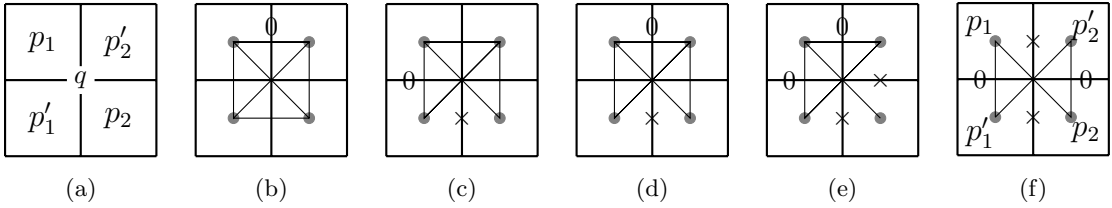
\begin{figure}[htb]
    \centering
    \begin{subfigure}[b]{0.16\textwidth}
       \begin{tikzpicture}
         \draw[step=1cm,black,thick] (0,0) grid (2,2);
            \filldraw [white] 
            (1,1) circle [radius=5pt];
            \draw(1,1) node{$q$};
            \draw(0.5,0.5) node{\large $p_1'$};
            \draw(0.5,1.5) node{\large $p_1$};
            \draw(1.5,0.5) node{\large $p_2$};
            \draw(1.5,1.5) node{\large $p_2'$};
        \end{tikzpicture}
        \centering\caption{}
    \end{subfigure}\hfill
    \begin{subfigure}[b]{0.16\textwidth}
    \begin{tikzpicture}
            \draw[step=1cm,black,thick] (0,0) grid (2,2);
            \filldraw [gray] 
            (0.5,0.5) circle [radius=2pt]
            (0.5,1.5) circle [radius=2pt]
            (1.5,0.5) circle [radius=2pt]
            (1.5,1.5) circle [radius=2pt];
            \draw (0.5,0.5) -- (1.5,0.5) -- (1.5,1.5) -- (0.5,1.5) -- (0.5,0.5) -- (1.5,1.5) -- (0.5,1.5) -- (1.5,0.5);
            \draw(1,1.7) node{$0$};
        \end{tikzpicture}
        \centering\caption{ }
    \end{subfigure}\hfill
    \begin{subfigure}[b]{0.16\textwidth}
       \begin{tikzpicture}
            \draw[step=1cm,black,thick] (0,0) grid (2,2);
            \filldraw [gray] 
            (0.5,0.5) circle [radius=2pt]
            (0.5,1.5) circle [radius=2pt]
            (1.5,0.5) circle [radius=2pt]
            (1.5,1.5) circle [radius=2pt];
            \draw (0.5,0.5) -- (1.5,1.5) -- (0.5,1.5) -- (0.5,0.5) -- (1.5,1.5) -- (0.5,1.5) -- (1.5,0.5) -- (1.5,1.5); 
            \draw(0.3,1) node{$0$};
            \draw(1,0.5) node{$\times$};
        \end{tikzpicture}     \centering\caption{}
    \end{subfigure}
    \begin{subfigure}[b]{0.16\textwidth}
       \begin{tikzpicture}
            \draw[step=1cm,black,thick] (0,0) grid (2,2);
            \filldraw [gray] 
            (0.5,0.5) circle [radius=2pt]
            (0.5,1.5) circle [radius=2pt]
            (1.5,0.5) circle [radius=2pt]
            (1.5,1.5) circle [radius=2pt];
            \draw (0.5,0.5) -- (1.5,1.5) -- (0.5,1.5) -- (0.5,0.5) -- (1.5,1.5) -- (0.5,1.5) -- (1.5,0.5) -- (1.5,1.5); 
            \draw(1,1.7) node{$0$};
            \draw(1,0.5) node{$\times$};
        \end{tikzpicture}     \centering\caption{}
    \end{subfigure}\hfill
        \begin{subfigure}[b]{0.16\textwidth}
       \begin{tikzpicture}
            \draw[step=1cm,black,thick] (0,0) grid (2,2);
            \filldraw [gray] 
            (0.5,0.5) circle [radius=2pt]
            (0.5,1.5) circle [radius=2pt]
            (1.5,0.5) circle [radius=2pt]
            (1.5,1.5) circle [radius=2pt];
            \draw (0.5,0.5) -- (1.5,1.5) -- (0.5,1.5) -- (0.5,0.5) -- (1.5,1.5) -- (0.5,1.5) -- (1.5,0.5); 
            \draw(1,1.7) node{$0$};
            \draw(0.3,1) node{$0$};    
            \draw(1,0.5) node{$\times$};
            \draw(1.5,1) node{$\times$};
        \end{tikzpicture}     \centering\caption{}
        \end{subfigure}\hfill
                \begin{subfigure}[b]{0.16\textwidth}
       \begin{tikzpicture}
            \draw[step=1cm,black,thick] (0,0) grid (2,2);
            \filldraw [gray] 
            (0.5,0.5) circle [radius=2pt]
            (0.5,1.5) circle [radius=2pt]
            (1.5,0.5) circle [radius=2pt]
            (1.5,1.5) circle [radius=2pt];
            \draw (0.5,0.5) -- (1.5,1.5) -- (1.5,0.5) -- (0.5,1.5) -- (0.5,0.5); 
            \draw(0.3,1) node{$0$};
            \draw(1.7,1) node{$0$};    
            \draw(1,0.5) node{$\times$};
            \draw(1,1.5) node{$\times$};

            \draw(0.3,0.3) node{\large $p_1'$};
            \draw(0.3,1.7) node{\large $p_1$};
            \draw(1.7,0.3) node{\large $p_2$};
            \draw(1.7,1.7) node{\large $p_2'$};
        \end{tikzpicture}     \centering\caption{}
        \end{subfigure}
        	\caption{Relationships of factors on $q$. The classification is organized in the increasing order of $\#_q \times$. (b) Case where $\#_q\times=0,\#_q \underline{0}=1$. The cases where $\#_q\times=0,\#_q \underline{0}\geq 1$ can be handled in the same way.  (c)(d) Case when $\#_q\times=1,\#_q \underline{0}= 1.$ Cases when $\#_q\times=1,\#_q \underline{0}\geq 1$ can be handled in the same way. (e)(f) Cases when  $\#_q \times=2,\#_q \underline{0}=2.$  }\label{fig:pattern}
\end{figure}

\noindent\textbf{Claim 2: The conditions (1)(2)(3) in  Lemma \ref{lem:all_regular} hold.}
Consider arbitrary qudit $q$,
let $\cA$ be the $C^*$-algebra generated by $\{p^q\}_{p\in\{p_1,p_2,p_1',p_2'\}}$. 
Since $q$ is not separable,
by lemma \ref{lem:C}, we know the
the center of $\cA$, i.e.$\cZ(\cA)$, must be trivial. Besides,  lemma \ref{lem:C} also implies $\cH^q=\cH_1\otimes \cH_2$, $\cA=\cL(\cH_1)\otimes \cI_{\cH_2}$. Since $q$ is not separable, we must further have $\cH_2$ is of dimension 
$1$, thus $\cA=\cL(\cH^q)$ (condition (3)).
Otherwise $q$ is again separable by considering $\cH^q = \bigoplus_i \cH_1\otimes \ket{\phi_i}\bra{\phi_i}$ where $\{\phi_i\}$ is a basis of $\cH_2$. 

Since Claim 1 is true, all terms are commuting in a regular way. This means all factors are either anti-commuting or commuting (condition (2)), thus $\forall$ term $p$, $(p^q)^2$ commute with every factor, i.e. $(p^q)^2\in \cZ(\cA)$. Since we already argued that $\cZ(\cA)$ must be trivial, we have $(p^q)^2=c_{pq}I$ for some constant $c_{pq}$ (condition (1)).
\end{proof}

\cite{bravyi2003commutative} showed 
that for qudit-factorized-CHP\footnote{Recall that CHP represents commuting Hamiltonian problem where the Hamiltonian might not be local.}, if all terms commute in a regular way, then one can transform the Hamiltonian into a qubit stabilizer Hamiltonian. Especially, they proved the following lemma.
\begin{lemma}[Lemma 12 of~\cite{bravyi2003commutative}]\label{lem:stabilizer}
    Let $\cH^q_*$ be a Hilbert space, $G_1,...,G_r\in\cL(\cH^q_*)$ be Hermitian operators such that 
    $$
    G_a^2=I,G_aG_b=\pm G_bG_a, \forall a,b\in\{1,...,r\}.
    $$
    and such that the algebra generated by $G_1,...,G_r$ coincides with $\cL(\cH^q_*)$. Then there exists an integer $n$, a tensor product structure $\cH^q_*=(\bC^2)^{\otimes n}$ and a unitary operator $U^q\in \cL(\cH^q_*)$ such that $U^qG_aU^{q\dagger}$ is a tensor of Pauli operators and the identity (up to sign) for all $a\in\{1,...,r\}$. Here $n$ may be equal to $0$.
\end{lemma}

Finally, we are prepared to prove the main theorem in this section.
\begin{theorem}\label{thm:weaker}
Consider a qudit-CLHP-2D-factorized Hamiltonian $H=\sum_p p$ on $n$-qudit space. If there are no separable qudits, then $H$ is equivalent to qubit stabilizer Hamiltonian on $\cH$.
\end{theorem}

\begin{proof}
 To match the notations in Definition \ref{def:equi_stab}, we denote the Hilbert space of qudit $q$ as $\cH_*^q$, the $n$-qudit space as $\cH_*:=\otimes_q \cH_*^q$. 
Consider a qudit-CLHP-2D-factorized $H=\sum_p p$. W.o.l.g assume that all $p$ are non-zero.
Note that since $p$ is Hermitian, if $p^2=0$, then $p=0$. 
Consider arbitrary qudit $q$,
use notations in Lemma \ref{lem:all_regular} we know $c_{pq}\neq 0, \forall p,q$.
 For every $p$, define $\tilde{p}$ be a normalized version of $p$, that is $\tilde{p}=\otimes_q \tilde{p}^q$ where $\tilde{p}^q =p^q/c_{pq}$. For any qudit $q$, view $...,\tilde{p}^q,...$ as $...,G_a,...$. By Lemma \ref{lem:all_regular}, one can check that $\{\tilde{p}^q\}_p$ satisfies the condition of Lemma \ref{lem:stabilizer}. Thus by choosing appropriate basis of each qudit, $\cH^q_*$ will have a factorized structure as $({\bC^2})^{\otimes m^q}$ for integer $m^q$, and $\{\tilde{p}^q\}$ are tensor of Pauli operators up to sign. 
Thus $p^q$ acts as a Pauli operator on this basis, up to phases. Thus $H$ is equivalent to qubit stabilizer Hamiltonian $H$ by definition.

\end{proof}

\subsection{The full version}\label{sec:strong}

Finally, we remove the constraints of no separable qudits in Theorem \ref{thm:weaker} and prove the following.   
\begin{theorem}\label{thm:stronger}
Any qudit-CLHP-2D-factorized Hamiltonian $H=\sum_p p$ is equivalent to a direct sum of qubit stabilizer Hamiltonian.\end{theorem}
\begin{proof}
Denote the space that $H$ acting on as $\cH=\bigotimes_q\cH^q$.
    Recall that if a qudit $q$ is separable with respect to decomposition $\cH^q=\bigoplus_i \cH_i^q$, for any chosen index $i$ we can restrict all terms on this subspace, and get a new instance of qudit-CLHP-2D-factorized. 
    If we repetitively perform this restriction whenever this is a separable qudit, after polytime  we will reach the case with no separable qudits. 
   
    To prove theorem \ref{thm:stronger}, it suffices to imagine the
restricting process as a decision tree. Specifically, we write down a root node and define the space of the root node  as $\cH$, and repeat the following process: Transverse all the  leaf nodes. Denote the leaf node considered currently as $*$, and its space as $\cH_*$. If $H$ restricting on $\cH_*$ has separable qudits, choose an arbitrary such separable qudit. Denote this qudit as $q$ and the corresponding decomposition as  $\cH^q=\bigoplus_i\cH^q_i$.
For every $i$, we build a child node $w_i$ to $*$, and define the space of $w_i$ as restricting $\cH^q$ to $\cH_i^q$ in $\cH_*$. We repeat this process until for every leaf node, there are no separable qudits.  
In the final tree, every leaf node $*$ corresponds to a simple subspace $\cH_*$. By the definition of the tree, we know $\{\cH_*\}_*$ are orthogonal to each other and $\cH=\bigoplus_* \cH_*$.
By Theorem \ref{thm:weaker},  $\cH$ is equivalent to qubit stabilizer Hamiltonian on $\cH_*$. Thus we prove that  $H$ is equivalent to a direct sum of qubit stabilizer Hamiltonian.
\end{proof}

\begin{corollary}
	Qudit-CLHP-2D-factorized is in $\NP$.
\end{corollary}
\begin{proof}
	Consider a qudit-CLHP-2D-factorized problem with Hamiltonian $H=\sum_{p} p$ and parameters $a,b$. By Theorem \ref{thm:stronger} we know there exists a $*$, where $\cH_*=\bigotimes_q \cH^q_*$ is a simple subspace, such that there is a ground state lies in $\cH_*$, and $H$ is equivalent to qubit stabilizer Hamiltonian on $\cH_*$. The $\NP$ prover is supposed to provide the subspace $\{\cH^q_*\}_q$, and provide the qudit unitary $U^q$ in Theorem \ref{lem:stabilizer} for each $q$. Using that information, the verifier firstly checks that all terms $\{p\}_p$ keep the subspaces $\{\cH^q_*\}_q$ invariant. Then the verifier
	uses polynomial time to  transform $\cH^q_*$ to be tensor of qubit space, i.e.$(\bC^2)^{\otimes m_q}$, and transform $H$ on $\cH_*$ to be a summation of Pauli operators up to phases, denoted as $H|_*=\sum_h a_hh$. Here $h$ is a Pauli operator.
	
	Then the verifier is going to verify $\lambda(H|_*)\leq a$.
	Since $\{a_hh\}_h$ are commuting, there is a ground state which is the common eigenstate of every $h$, denote the corresponding eigenvalue as $\lambda_{h}$. The prover is supposed to provide such $\{\lambda_h\}_h$. The verifier verifies that $\sum_h a_h\lambda_h\leq a$, and verifies there is a state which is the common 1-eigenstate of commuting Pauli operators $\{h/\lambda_h\}$. The common 1-eigenstate verification can be done in polynomial time by standard stabilizer formalism. Note that although we describe the prover in an interactive way, he can in fact send all the witnesses at the same time. 
	\end{proof}
\appendix

\section{Acknowledgement}
Thanks Thomas Vidick and Daniel Ranard for helpful discussion.

\section{Relationship between general case and projection case}\label{appendix:equi_proj}

\begin{lemma}\label{lem:qudit_proj_redu}
    If qudit-CLHP-2D-projection is in $\NP$, then qudit-CLHP-2D is in $\NP$.
\end{lemma}
\begin{proof}
	Consider a qudit-CLHP-2D$(H,a,b)$, where $H=\sum_p p$, $a,b\in\bR$ and $b-a\geq 1/poly(n)$. Denote the ground energy $\lambda:=\lambda(H)$.	 Since $\{p\}_p$ are commuting with each other, there exists a ground state $\ket{\psi}$ and $\{\lambda_p\in\bR\}_p$ such that $p\ket{\psi}=\lambda_p \ket{\psi}$,$\forall p$ and $\sum_p \lambda_p=\lambda$. Let $\Pi_p$ be the projection onto the $\lambda_p$-eigenspace of $p$. Let $\hat{p}=I-\Pi_p$,
	 $\hat{H}=\sum_p \hat{p}$. Since $\{p\}_p$ are commuting, we know that $\{\hat{p}\}_p$ are also commuting.

	 The $\NP$ prover is supposed to  list such $\{\lambda_p\}_p$, then the verifier can check 
$\sum_p \lambda_p<a$ and 	 
	 compute $\hat{p}$ and $\hat{H}$. Then the prover is supposed to prove	  qudit-CLHP-2D-projection$(\hat{H})$ is a Yes instance --- that is proving
	 there  exists $\ket{\psi}$ such that $\hat{p} \ket{\psi}=0,\forall \hat{p}$. Thus if qudit-CLHP-2D-projection is in $\NP$, then qudit-CLHP-2D is in $\NP$.
\end{proof}

\begin{corollary}\label{cor:proj_redu}
	  If  the qutrit-CLHP-2D-projection is in $\NP$, then  the qutrit-CLHP-2D is in $\NP$.
\end{corollary}

\section{Qudits on the vertexs or on the edges}\label{appendix:edge_vertex}

In this manuscript, we consider commuting local Hamiltonian on a 2D square lattice, where qudits are  on the vertices and Hermitian terms are on the plaquettes. There is another setting that put qudits on the edges, 
 and Hermitian terms on "plaquettes" and "stars". Specifically, 
   as shown in Fig.~\ref{fig:edge}
 for each plaquette $p$, there is a Hermitian term $B_p$ acting on the qudits on its edges, i.e. $q_1,q_2,q_3,q_4$. For each vertex $v$, consider the star consisting  of $v$ and edges adjacent to $v$, there is a Hermitian term $A_v$ acting on qudits on its edges, i.e. $q_3,q_4,q_5,q_6$.  The Hamiltonian is $H=\sum_p B_p +\sum_v A_v$. We abbreviate this setting as ``qudits on 2D edges" and the setting in Sec.~\ref{sec:prelim} as ``qudits on 2D vertices". In the following we will show that the two settings are equivalent.
 
 \begin{figure}[!ht]
	\centering
 \begin{tikzpicture}
\draw[step=1cm,black,thick] (0,0) grid (2,2);

\filldraw [black] 
(0.5,0) circle [radius=1.5pt]
(1.5,0) circle [radius=1.5pt]
(0,0.5) circle [radius=1.5pt]
(1,0.5) circle [radius=1.5pt]
(2,0.5) circle [radius=1.5pt]
(0.5,1) circle [radius=1.5pt]
(1.5,1) circle [radius=1.5pt]

(0,1.5) circle [radius=1.5pt]
(1,1.5) circle [radius=1.5pt]
(2,1.5) circle [radius=1.5pt]
(0.5,2) circle [radius=1.5pt]
(1.5,2) circle [radius=1.5pt];

\filldraw[white] 
(1,1) circle [radius=5pt];

\filldraw[white] 
(0,0.5) circle [radius=5pt];
\filldraw[white] 
(0.5,0) circle [radius=5pt];
\filldraw[white] 
(1,0.5) circle [radius=5pt];
\filldraw[white] 
(0.5,1) circle [radius=5pt];
\filldraw[white] 
(1.5,1) circle [radius=5pt];
\filldraw[white] 
(1,1.5) circle [radius=5pt];

\draw(0,0.5) node{$q_1$};
\draw(0.5,0) node{$q_2$};
\draw(1,0.5) node{$q_3$};
\draw(0.5,1) node{$q_4$};

\draw((1.5,1) node{$q_5$};
\draw(1,1.5) node{$q_6$};

\draw(0.5,0.5) node{\large $p$};
\draw(1,1) node{\large $v$};

\end{tikzpicture}
	\caption{Qudits on edges to qudits on vertices}\label{fig:edge}
\end{figure}
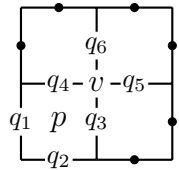

 \noindent\textbf{(1) ``qudits on 2D edges" $\Rightarrow$ ``qudits on 2D vertices".} Begin from ``qudits on 2D edges",  as shown in Fig.~\ref{fig:edge_vertex}, the qudits on the edges can in fact be viewed as qudits on the vertices of another 2D square lattice defined by the dash lines. The terms $B_p$ and $A_s$ will correspond to plaquette terms in the dashed lattice.  Thus our techniques directly apply to the setting for qudits on the edges.

\begin{figure}[!ht]
	\centering
 \begin{tikzpicture}
\draw[step=1cm,black,thick] (0,0) grid (2,2);

\filldraw [black] 
(0.5,0) circle [radius=1.5pt]
(1.5,0) circle [radius=1.5pt]
(0,0.5) circle [radius=1.5pt]
(1,0.5) circle [radius=1.5pt]
(2,0.5) circle [radius=1.5pt]
(0.5,1) circle [radius=1.5pt]
(1.5,1) circle [radius=1.5pt]

(0,1.5) circle [radius=1.5pt]
(1,1.5) circle [radius=1.5pt]
(2,1.5) circle [radius=1.5pt]
(0.5,2) circle [radius=1.5pt]
(1.5,2) circle [radius=1.5pt];

\draw[dashed] (-0.5,1) -- (1,2.5);
\draw[dashed] (0,0.5) -- (1.5,2);
\draw[dashed] (0.5,0) -- (2,1.5);
\draw[dashed] (1,-0.5) -- (2.5,1);

\draw[dashed] (-0.5,1) -- (1,-0.5);
\draw[dashed] (0,1.5) -- (1.5,0);
\draw[dashed] (0.5,2) -- (2,0.5);
\draw[dashed] (1,2.5) -- (2.5,1);
\end{tikzpicture}
	\caption{Qudits on 2D edges to qudits on 2D vertices}\label{fig:edge_vertex}
\end{figure}

\noindent\textbf{(1) ``qudits on 2D vertices" $\Rightarrow$ ``qudits on 2D edges".} Begin from ``qudits on 2D vertices",  as shown in Fig.~\ref{fig:vertex_edge}, 
the qudits on the vertices can in fact be viewed as qudits on the edges of another 2D square lattice defined by the dash lines. The plaquette  terms will correspond to plaquette and star terms in the dashed lattice. 

\begin{figure}[!ht]
	\centering
\scalebox{0.6}{ \begin{tikzpicture}
\draw[step=1cm,black,thick]  (0,0) grid (7,7);
scriptsize
\foreach \i in {2,...,5}
{
\foreach \j in {2,...,5}
{
\filldraw [black] 
(\i,\j) circle [radius=3pt]; }      
     };
     
\foreach \i in {1,3,5}
{
\draw[dashed] (\i,0) -- (7,7-\i);
\draw[dashed] (0,\i) -- (7-\i,7);  
     };

\foreach \i in {2,4,6}
{
\draw[dashed] (0,\i) -- (\i,0); 
\draw[dashed] (7-\i,7) -- (7,7-\i);     
     };

\end{tikzpicture}}
	\caption{Qudits on 2D edges to qudits on 2D vertices}\label{fig:vertex_edge}
\end{figure}



\bibliographystyle{alpha}
\bibliography{ref.bib}

\end{document}